\documentclass{article}
\usepackage{fullpage}

\usepackage[utf8]{inputenc}
\usepackage{amsthm,amssymb}
\usepackage{amsmath}
\usepackage{thmtools,mathtools}
\usepackage{caption,subcaption}
\usepackage{epsfig,graphicx,graphics,color,tikz,tikz-cd}
\usepackage [linesnumbered, ruled,vlined]{algorithm2e}
\usetikzlibrary{calc,decorations.pathmorphing,shapes}
\usepackage[most]{tcolorbox}
\usepackage{array,makecell,multirow}
\setcellgapes{4pt}
\usepackage{bbm}
\usepackage{varwidth}
\usepackage{rotating}
\usepackage{booktabs}
\usepackage[mathcal,mathscr]{eucal}
\usepackage{varwidth}
\usepackage{enumerate,enumitem}
\usepackage[nocompress,sort]{cite}
\usepackage{xspace}
\usepackage{bm}
\usepackage{csquotes}
\usepackage{hyperref}
\hypersetup{
colorlinks=true,       
linkcolor=blue,        
citecolor=magenta,         
filecolor=magenta,     
urlcolor=cyan,         
linktoc=all
}

\theoremstyle{plain}
\newtheorem{thm}{Theorem}[section]
\newtheorem{lem}[thm]{Lemma}

\newtheorem{cl}[thm]{Claim}

\theoremstyle{definition}

\newtheorem{rem}[thm]{Remark}
\newtheorem{case}{Case}
\newtheorem{subcase}{Subcase}[case]

\usepackage{etoolbox}
\AtBeginEnvironment{thm}{\setcounter{case}{0}}
\AtBeginEnvironment{lem}{\setcounter{case}{0}}

\def\final{0}  
\def\iflong{\iffalse}
\ifnum\final=0  

\else 
\newcommand{\knote}[1]{}
\newcommand{\ynote}[1]{}
\fi

\makeatletter
\renewcommand{\paragraph}{%
  \@startsection{paragraph}{4}%
  {\z@}{1.5ex \@plus 1ex \@minus .2ex}{-1em}%
  {\normalfont\normalsize\bfseries}%
}
\makeatother

\newcommand{\bZ}{\mathbb{Z}}

\newcommand{\cG}{\mathcal{G}}

\newcommand{\cI}{\mathcal{I}}

\newcommand{\cM}{{\bm{M}}}

\newcommand{\maxpt}{\textsc{Max-PT}\xspace}
\newcommand{\pst}{\textsc{PST}\xspace}
\newcommand{\opt}{\textsc{Opt}\xspace}

\definecolor[named]{Blue}{cmyk}{1,0.1,0,0.1}
\definecolor[named]{Yellow}{cmyk}{0,0.16,1,0}
\definecolor[named]{Orange}{cmyk}{0,0.42,1,0.01}
\definecolor[named]{Red}{cmyk}{0,0.90,0.86,0}
\definecolor[named]{LightBlue}{cmyk}{0.49,0.01,0,0}
\definecolor[named]{Green}{cmyk}{0.20,0,1,0.19}

\SetKwFor{Whileblock}{while}{do}{}
\SetKwBlock{Begin}{}{}
\SetVlineSkip{1.5pt}
\makeatletter
\newcommand{\linkdest}[1]{\Hy@raisedlink{\hypertarget{#1}{}}}
\makeatother

\newlength{\bibitemsep}
\newlength{\bibparskip}
\let\oldthebibliography\thebibliography
\renewcommand\thebibliography[1]{%
\oldthebibliography{#1}%
\setlength{\parskip}{\bibitemsep}%
\setlength{\itemsep}{\bibparskip}%
}

\usepackage{titling}
\thanksmarkseries{alph}

\title{Above-Guarantee Algorithm
for Properly Colored Trees}

\author{
Yuhang Bai\thanks{School of Mathematics and Statistics, Northwestern Polytechnical University and Xi'an-Budapest Joint Research Center for Combinatorics, Xi'an 710129,
Shaanxi, People's Republic of China. Email: \texttt{yhbai@mail.nwpu.edu.cn}.}
\and
Kristóf Bérczi\thanks{MTA-ELTE Matroid Optimization Research Group and HUN-REN–ELTE Egerváry Research Group, Department of Operations Research, ELTE Eötvös Loránd University, and HUN-REN Alfréd Rényi Institute of Mathematics, Budapest, Hungary. Email: \texttt{kristof.berczi@ttk.elte.hu}.}
}

\date{}

\begin{document}

\maketitle

\begin{abstract}
In the \textit{Properly Colored Spanning Tree} problem, we are given an edge-colored undirected graph and the goal is to find a spanning tree in which any two adjacent edges have distinct colors. Since finding such a tree is NP-hard in general, previous work often relied on minimum color degree conditions to guarantee the existence of properly colored spanning trees. While it is known that every connected edge-colored graph $G$ contains a properly colored tree of order at least $\min\{|V(G)|, 2\delta^c(G)\}$, where $\delta^c(G)$ denotes the minimum number of colors incident to a vertex, we study the algorithmic above-guarantee problem for properly colored trees. We provide a polynomial-time algorithm that constructs a properly colored tree of order at least $\min\{|V(G)|, 2\delta^c(G)+1\}$ in a connected edge-colored graph $G$, whenever such a tree exists.

\medskip

\end{abstract}

\section{Introduction}

All graphs considered in this paper are finite and simple. 
A \emph{$k$-edge-colored graph} is a graph $G=(V,E)$ with a coloring $c\colon E\to [k]$ of its edges by $k$ colors.
For a vertex $v\in V$, we use $N_G(v)$ to denote the set of all neighbors of $v$ in $G$.
The \emph{color degree} of $v$, denoted by $d_G^c(v) = \bigl|\{c(vu)\colon u\in N_G(v)\}\bigr|$.
The \emph{minimum color degree} of $G$ is $\delta^c(G) = \min_{v\in V} d_G^c(v)$.
We refer to a graph that is $k$-edge-colored for some $k\in\bZ_+$ as \emph{edge-colored}.
A subgraph $H$ of $G$ is called \emph{rainbow} if no two edges of $H$ have the same color, and \emph{properly colored} if any two adjacent edges of $H$ have distinct colors.

Degree conditions for spanning structures are a central topic in extremal graph theory. A classical example is Dirac's theorem~\cite{dirac1952some} on Hamiltonian cycles. Recently, Fomin, Golovach, Sagunov, and Simonov~\cite{fomin2022algorithmic} studied an algorithmic extension of Dirac's theorem. Specifically, they showed that deciding whether a 2-connected graph $G$ contains a cycle of length at least $\min\{|V(G)|, 2\delta(G) + k\}$ is fixed-parameter tractable when parameterized by $k$. This framework is referred to as an \textit{above-guarantee problem} in parameterized complexity. The general idea of this paradigm is that the natural parameterization of a maximization problem by its solution size is often unsatisfactory if a sufficiently large lower bound on that size is already guaranteed.

Building on this success, it is natural to ask whether similar guarantees exist in edge-colored graphs. Since a rainbow forest is a common independent set of two matroids, i.e., the partition matroid defined by the color classes and the graphic matroid of the graph, a rainbow forest of maximum size can be found in polynomial time using Edmonds' celebrated matroid intersection algorithm~\cite{edmonds1970submodular}. However, much less is known about the properly colored case. In~\cite{borozan2019maximum}, Borozan, de La Vega, Manoussakis, Martinhon, Muthu, Pham, and Saad initiated the study of properly colored spanning trees of edge-colored graphs and introduced the \emph{Properly Colored Spanning Tree} problem (\pst). 
A simple motivation for properly colored trees is local conflict avoidance. If edge colors represent resources such as frequencies, channels, or labels, then a properly colored spanning tree is a connected backbone in which no vertex uses the same resource on two selected incident edges. 
This problem generalizes the well-known bounded-degree spanning tree problem for uncolored graphs, since the number of colors effectively limits the degree of each vertex. It also generalizes the properly colored Hamiltonian path problem when the number of colors is restricted to two. Since both of these problems are NP-complete, finding a properly colored spanning tree is computationally hard in general. 

To overcome this computational hardness, we use the minimum color degree to guarantee the existence of large properly colored structures.
Hu, Li, and Maezawa~\cite{hu2022maximum} proved that every connected edge-colored graph $G$ contains a properly colored tree of order at least $\min\{|V(G)|, 2\delta^c(G)\}$.
This bound is tight, since some connected edge-colored graphs do not contain a properly colored tree of order $\min\{|V(G)|, 2\delta^c(G)+1\}$.
They also characterized the corresponding families of extremal graphs.
However, their proof does not yield a constructive algorithm.
Motivated by the work of Fomin et al.~\cite{fomin2022algorithmic}, we study the algorithmic above-guarantee problem: 
given a connected edge-colored graph $G$, can we decide in polynomial time whether it contains a properly colored tree of order at least $\min\{|V(G)|, 2\delta^c(G)+1\}$, and output one if it exists.

\subsection{Related work and motivation}
\label{sec:related}

Finding properly colored spanning trees in graphs is closely related to constrained spanning tree problems, or in a more general context, to the problem of finding a basis of a matroid subject to further matroid constraints. In what follows, we give an overview of questions that motivated our investigations. 

\paragraph{Properly colored trees}

Properly colored spanning trees were first considered in \cite{borozan2019maximum} where their existence was studied from both a graph-theoretic and an algorithmic perspective.
They showed that the problem of finding a properly colored spanning tree remains NP-complete even for complete graphs.
Deciding the existence of a properly colored spanning tree is hard in general, hence, a considerable amount of work has focused on finding sufficient conditions~\cite{cheng2020properly,kano2020color,kano2021rainbow}.
Since a properly colored spanning tree may not exist, it is natural to ask for the maximum size of a properly colored tree not necessarily spanning all the vertices, called the \emph{Maximum-size Properly Colored Tree} problem (\maxpt).
Borozan et al.~\cite{borozan2019maximum} proved that \maxpt is hard to approximate within a factor of $55/56+\varepsilon$ for any $\varepsilon>0$, although they provided polynomial-time algorithms for graphs not containing properly colored cycles.
Recently, Bai, Bérczi, Csáji, and Schwarcz~\cite{bai2026approximating} showed that \maxpt is NP-hard to approximate within a factor of $1/n^{1-\varepsilon}$ for any $\varepsilon >0$, even for instances containing a properly colored spanning tree.
Furthermore, for any fixed constant $\varepsilon >0$, they provided a polynomial-time algorithm that achieves a $1/\sqrt{(2+\varepsilon)(n-1)}$-approximation for \maxpt on complete multigraphs on $n$ vertices.
Hu et al.~\cite{hu2022maximum} proved that the maximum order of a properly colored tree in an edge-colored connected graph $G$ is at least $\min \{|V(G)|, 2 \delta^c(G)\}$.

\paragraph{Above-guarantee problems}
The study of parameterized problems above or below guarantees is a well-established paradigm in parameterized complexity; see the survey by Gutin and Mnich~\cite{gutin2025survey}.
Following this line of research, Fomin et al.~\cite{fomin2022algorithmic} gave an FPT algorithm, parameterized by $k$, to decide whether a 2-connected graph $G$ contains a cycle of length at least $\min\{|V(G)|, 2\delta(G)+k\}$.
Later, Fomin, Golovach, Sagunov, and Simonov~\cite{fomin2024longest} provided an FPT algorithm, parameterized by $k$, to decide whether a 2-connected graph $G$ contains a cycle of length at least $\mathrm{mad}(G)+k$, where $\mathrm{mad}(G)$ denotes the maximum average degree of $G$.
In addition, the same authors~\cite{fomin2024approxdirac} showed that given any polynomial-time algorithm finding a cycle of length $f(L)$ for a subadditive function $f$, one can construct a cycle of length at least $2\delta(G)+\Omega(f(L-2\delta(G)))$ in a 2-connected graph $G$.
For trees, they~\cite{fomin2025tree} gave a randomized FPT algorithm parameterized by $k$ to decide whether a graph contains a given tree of order at most $\delta(G)+k$.

Motivated by the extremal bound of Hu et al.~\cite{hu2022maximum} and the above-guarantee paradigm, we study properly colored trees beyond this bound. Since the bound is tight, finding a tree of order at least $\min\{|V(G)|, 2\delta^c(G)+1\}$ is the first nontrivial step above the guarantee. This task is challenging. Unlike the uncolored cycles considered in Dirac-type problems, we seek a tree subject to local color constraints. Moreover, beyond the strong inapproximability of \maxpt discussed earlier, even finding a maximum properly colored forest is a nontrivial special case of matroid 3-parity; see~\cite{bai2025approximating}. While above-guarantee problems for uncolored cycles are fixed-parameter tractable, achieving the threshold $\min\{|V(G)|, 2\delta^c(G)+1\}$ in our edge-colored setting requires new algorithmic ideas.

\subsection{Our results}
\label{sec:results}

Our contribution is twofold. First, we show that finding a maximum-size rainbow tree is NP-hard even in connected star-colored graphs (Theorem~\ref{thm:NP}). Second, we give a polynomial-time algorithm that constructs a rainbow tree of order at least $\min\{|V(G)|, 2\delta^c(G)+1\}$ in star-colored connected graphs, provided such a tree exists (Theorem~\ref{thm:rb-star}). We extend this result to general connected edge-colored graphs by giving a polynomial-time algorithm that decides whether a properly colored tree of order at least $\min\{|V(G)|, 2\delta^c(G)+1\}$ exists and outputs one if so (Theorem~\ref{thm:main}).

\paragraph{Paper organization.}

The paper is organized as follows.
In Section~\ref{sec:prelim}, we introduce basic definitions and notation, and overview some results that we will use in our proofs.
In Section~\ref{sec:hardness}, we discuss the complexity of the maximum-size rainbow tree problem in star-colored graphs.
Section~\ref{sec:algo} presents our main algorithms.

\section{Preliminaries}
\label{sec:prelim}

\paragraph{Basic notation.}

We denote the set of \emph{positive integers} by $\bZ_+$. For a positive integer $k$, we use $[k]\coloneqq\{1,\dots,k\}$. Given a ground set $S$, the \emph{difference} of $X,Y\subseteq S$ is denoted by $X\setminus Y$. If $Y$ consists of a single element $y$, then $X\setminus \{y\}$ and $X\cup \{y\}$ are abbreviated as $X-y$ and $X+y$, respectively. Similarly, the single element set $\{y\}$ is often denoted by $y$.

All graphs considered in this paper are finite and simple. 
For a graph $G$, we denote its vertex set and edge set by $V(G)$ and $E(G)$, respectively.
A graph is \emph{simple} if it has no parallel edges and no loops. A simple graph is \emph{complete} if it contains exactly one edge between any pair of vertices. 
For a graph $G$, we write $|G|=|V(G)|$ for its \emph{order}.
Let $G=(V,E)$ be a graph, and $X\subseteq V$ be a subset of vertices.
The \emph{subgraph of $G$ induced by $X$} is denoted by $G[X]$.
The \emph{graph obtained by deleting $X$} is denoted by $G-X$.
Denote $\partial(X)=\{xy\in E(G)\colon x\in X, y\notin X\}$ for vertex set $X\subseteq V$.

Let $c\colon E\to[k]$ be an edge-coloring of $G$ using $k$ colors.
The function $c$ is extended to subsets of edges where, for a subset $F\subseteq E$ of edges, $c(F)$ denotes the set of colors appearing on the edges in $F$.
For an edge-colored graph $G=(V,E)$, we use $E_i=\{ e\in E\mid c(e)=i\}$ to denote the set of edges of color $i$.
An edge-colored graph is called \emph{star-colored} if $E_i$ is a star for each color $i\in[k]$.

\paragraph{Extremal graphs.}

Hu et al.~\cite{hu2022maximum} introduced six families of edge-colored graphs, denoted $\cG^1, \dots, \cG^6$, based on six specific graphs, see Figure~\ref{fig:1} for an example. 
For completeness, we recall their definitions; we present the family $\cG^1$ in the tournament-based form that will be used in our arguments.
\begin{figure}[t!]
\centering
    \includegraphics[width=0.9\textwidth]{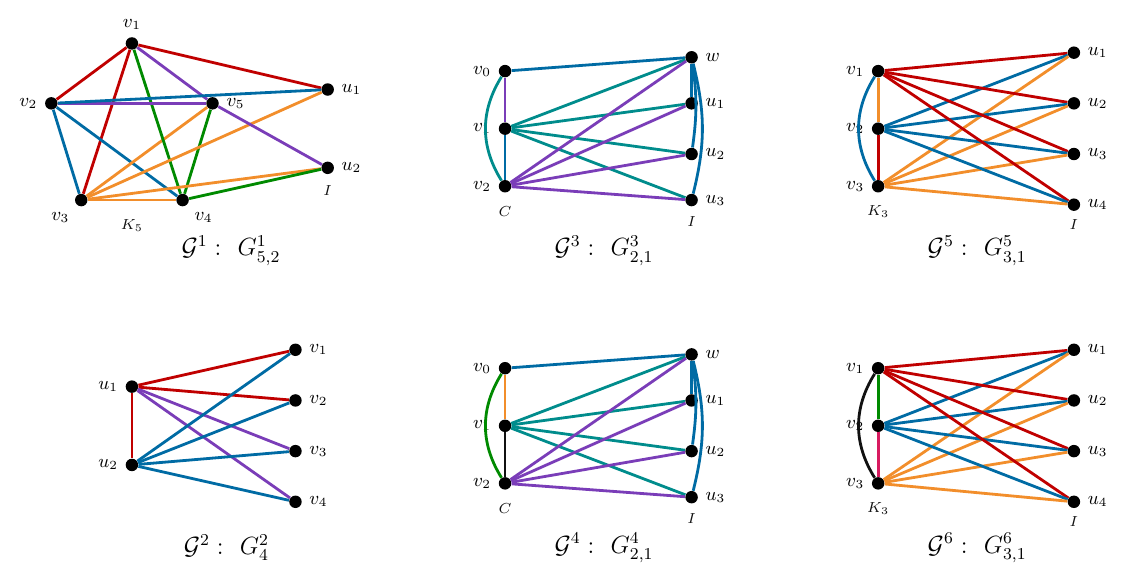}
\caption{Illustration of the extremal graphs.}
\label{fig:1}
\end{figure}

\medskip

\noindent \textit{Graph $G^1_{m,k}$.} Let $m \ge 3$ be odd and $k \ge 2$, and let $c_1, \dots, c_m$ be pairwise distinct colors. 
Let $K_m$ be an edge-colored complete graph on the vertex set $\{v_1,\ldots,v_m\}$ such that there exists a regular tournament $D$ on $V(K_m)$ satisfying $c(v_i v_j)=c_i$ whenever the arc $(v_i,v_j)$ belongs to $D$. 
Let $I = \{u_1, \dots, u_k\}$ be an independent set disjoint from $V(K_m)$. For each $j \in [k]$, choose a subset $U_j \subseteq V(K_m)$ with $|U_j| \ge (m+1)/2$, and for each $v_i \in U_j$, add the edge $v_i u_j$ with color $c_i$. The resulting graph is $G^1_{m,k}$. 

\medskip

\noindent \textit{Graph $G^2_k$.} Let $k \ge 3$ and vertices $u_1, u_2, v_1, \dots, v_k$ be distinct. Let $c_1, c'_1, c_2$ be pairwise distinct colors. Partition ${v_1, \dots, v_k}$ into sets $V_1$ and $V'_1$, with $V'_1 \neq \emptyset$. Add edges $u_2 v_i$ for all $i \in [k]$ with color $c_2$, edges $u_1 v_i$ for $v_i \in V_1$ with color $c_1$, and edges $u_1 v_i$ for $v_i \in V'_1$ with color $c'_1$. Finally, add the edge $u_1 u_2$ with color $c_1$. There are no other edges. The resulting graph is $G^2_k$.

\medskip

\noindent \textit{Graph $G^3_{m,k}$.} Let $m,k \ge 1$, $C = \{v_0, \dots, v_m\}$ induce a properly colored $K_{m+1}$, $w \notin C$, and $I = \{u_1, \dots, u_{m+k}\}$ be independent with $I \cap (C +w) = \emptyset$. Let $c, c_1, \dots, c_m$ be pairwise distinct colors, with $c$ not on edges incident to $v_0$ in $C$ and $c_i$ not on edges incident to $v_i$ in $C$. Add edges $v_0 w$ and $w u_j$ for $j \in [m+k]$ with color $c$. For each $i \in [m]$, add edges $v_i w$ and $v_i u_j$ for $j \in [m+k]$ with color $c_i$. No edges are added between $v_0$ and $I$. The resulting graph is $G^3_{m,k}$.

\medskip

\noindent\textit{Graph $G^4_{m,k}$.} Construct a graph as in the definition of $G^3_{m,k}$, except that $C$ is a rainbow complete graph whose edge colors are disjoint from $\{c, c_1, \dots, c_m\}$. The resulting graph is denoted by $G^4_{m,k}$.

\medskip

\noindent\textit{Graph $G^5_{m,k}$.} Let $m,k \ge 1$. Let $K_m$ be a properly colored complete graph on vertices $v_1, \dots, v_m$, and let $I = \{u_1, \dots, u_{m+k}\}$ be an independent set disjoint from $V(K_m)$. Let $c_1, \dots, c_m$ be pairwise distinct colors, each absent from edges incident to the corresponding vertex $v_i$ in $K_m$. For every $i \in [m]$ and $j \in [m+k]$, add the edge $v_i u_j$ with color $c_i$. The resulting graph is denoted by $G^5_{m,k}$.

\medskip

\noindent \textit{Graph $G^6_{m,k}$.} Construct the graph as in the definition of $G^5_{m,k}$, except that the complete graph on $v_1, \dots, v_m$ is rainbow and its edge colors are disjoint from $\{c_1, \dots, c_m\}$. The resulting graph is denoted by $G^6_{m,k}$.

\medskip

Finally, the graph families are defined as
\[
\cG^1 = \{G \colon G \cong G^1_{m,k} \text{ for some odd } m \ge 3 \text{ and } k \ge 2\},
\]
\[
\cG^2 = \{G \colon G \cong G^2_k \text{ for some } k \ge 3\},
\]
and for each $i \in \{3,4,5,6\}$,
\[
\cG^i = \{G \colon G \cong G^i_{m,k} \text{ for some } m \ge 1 \text{ and } k \ge 1\}.
\]

For later use, we note that $\delta^c(G^1_{m,k})=(m+1)/2$, $\delta^c(G^2_k)=2$, $\delta^c(G^3_{m,k})=\delta^c(G^4_{m,k})=m+1$, and $\delta^c(G^5_{m,k})=\delta^c(G^6_{m,k})=m$.

\begin{rem}\label{re:extre-test}
The graph in each of the families $\mathcal G^1,\ldots,\mathcal G^6$ can be check in polynomial time.
For $\mathcal G^2,\ldots,\mathcal G^6$, the underlying uncolored graph is a
complete split graph. Thus the required partition is determined, except in $\mathcal G^3$ and $\mathcal G^4$, where we additionally try all $O(n^2)$ choices of the special vertices $w$ and
$v_0$. For each candidate, all adjacency and color conditions in the definitions can be checked by scanning the edge set.

For $\mathcal G^1$, candidate split partitions $K\cup I$ can be enumerated in polynomial time, and the adjacency and degree conditions are then checked directly. The remaining tournament-coloring condition reduces to choosing, for each $v\in K$, one of at most two possible colors $\gamma(v)$. Introduce one Boolean variable for each vertex with two possible choices. Whenever a pair of choices for two vertices $u,v\in K$ violates either $\gamma(u)\ne \gamma(v)$ or $c(uv)\in{\gamma(u),\gamma(v)}$, we add the corresponding 2-SAT clause forbidding this pair. Hence the compatibility constraints form a polynomial-size 2-SAT instance. Once the colors $\gamma(v)$ are chosen, the orientation of every edge $uv\in E(K)$ is determined by the endpoint whose chosen color equals $c(uv)$, and regularity of the resulting tournament is checked directly. Thus membership in $\mathcal G^1$ can be tested in polynomial time.
\end{rem}

We will use the following known results.

\begin{thm}[Hu et al.~\cite{hu2022maximum}]\label{thm:PC-ex}
    Let $G$ be an edge-colored connected graph. The graph $G$ does not contain a properly colored tree of order at least $\min \left\{|G|, 2 \delta^{c}(G)+1\right\}$ if and only if $G \in \mathcal{G}^1 \cup \mathcal{G}^2 \cup \mathcal{G}^3 \cup \mathcal{G}^5$.
\end{thm}

\begin{thm}[Hu et al.~\cite{hu2022maximum}]\label{thm:Rb-ex}
    Let $G$ be a star-colored connected graph. 
    The graph  $G$ does not contain a rainbow tree of order at least $\min \left\{|G|, 2 \delta^{c}(G)+1\right\}$ if and only if $G \in \mathcal{G}^1 \cup \mathcal{G}^2 \cup \mathcal{G}^4 \cup \mathcal{G}^6$.
\end{thm}

\begin{thm}[Suzuki~\cite{suzuki2006necessary}]\label{thm:Rb}
    An edge-colored connected graph $G$ of order $n$ has a rainbow spanning tree if and only if, for every $r$ with $1 \le r \le n-2$, removing all edges of any $r$ colors from $G$ leaves a graph with at most $r+1$ components.
\end{thm}

\paragraph{Matroids.}

For basic definitions on matroids and on matroid optimization, we refer the reader to~\cite{oxley2011matroid,frank2011connections}.
A \emph{matroid} $\cM=(E,\cI)$ is defined by its \emph{ground set} $E$ and its \emph{family of independent sets} $\cI\subseteq 2^E$ that satisfies the \emph{independence axioms}: (I1) $\emptyset\in\cI$, (I2) $X\subseteq Y,\ Y\in\cI\Rightarrow X\in\cI$, and (I3) $X,Y\in\cI,\ |X|<|Y|\Rightarrow\exists e\in Y\setminus X\ s.t.\ X+e\in\cI$.
Members of $\cI$ are called \emph{independent}, while sets not in $\cI$ are called \emph{dependent}.

Let $\cM$ be a matroid on a ground set $E$ with rank function $r$.
For a subset $X \subseteq E$, the \emph{restriction} of $\cM$ to $X$, denoted by $\cM|X$, is the matroid on $X$ whose independent sets are exactly the independent sets of $\cM$ that are contained in $X$.
For an element $e \in E$, the \emph{contraction} of $\cM$ by $e$, denoted by $\cM/e$, is the matroid on the ground set $E -e$ whose rank function $r_{\cM/e}$ is defined by
$$
r_{\cM/e}(Y) = r(Y+e) - r(e)
$$
for any subset $Y \subseteq E -e$.
Equivalently, for $e$ that is not a loop, a subset $I \subseteq E-e$ is independent in $\cM/e$ if and only if $I+e$ is independent in $\cM$.
The intersection of two matroids $\cM_1$ and $\cM_2$ is denoted by $\cM_1\cap \cM_2$; note that this is not a matroid in general.

Let $\{E_1, \ldots, E_t\}$ be a partition of a finite set $E$, and let $g_1, \ldots, g_t$ be non-negative integers. A \emph{partition matroid} on $E$ is a matroid whose independent sets are the subsets $I \subseteq E$ satisfying $|I \cap E_i| \leq g_i$ for all $i \in [t]$. For a graph $G=(V,E)$, the \emph{graphic matroid} $\cM=(E,\cI)$ of $G$ is a matroid whose independent sets are the forests, that is, $\cI=\{F\subseteq E\colon F\ \text{is a forest}\}$. 

\section{Hardness results}
\label{sec:hardness}

This section is devoted to the hardness of the maximum-size rainbow tree problem. We show that the problem remains NP-hard even for star-colored graphs, via a polynomial-time reduction from MAX-SAT. Note that, in star-colored graphs, a properly colored tree is also a rainbow tree.

\begin{thm}\label{thm:NP}
    It is NP-hard to find a maximum-size rainbow tree in star-colored graphs.
\end{thm}

\begin{proof}
    We give a reduction from MAX-SAT by constructing a star-colored graph $G_\Phi$ from a CNF formula $\Phi$ with variables $x_1,\dots,x_s$ and clauses $c_1,\dots,c_t$ in polynomial time.

    For each variable $x_i$, we introduce three vertices $y_i, a_{i1}$, and $a_{i2}$. 
    For each clause $c_j$, we introduce a vertex $c'_j$. 
    For each $i$, we add the edges $y_i a_{i1}$ and $y_i a_{i2}$. 
    Next, for each clause $c_j$ and each variable $x_i$, we add the edge $c'_ja_{i1}$ if the positive literal $x_i$ occurs in $c_j$, and we add the edge $c'_ja_{i2}$ if the negative literal $\neg x_i$ occurs in $c_j$. 
    Finally, we add the edges $y_i y_{i+1}$ for all $i \in [s-1]$. We introduce three pairwise disjoint sets of colors: $\{b_i\colon i\in [s]\}$, $\{C_j\colon j\in [t]\}$, and $\{d_i\colon i\in [s-1]\}$. 
    For each $i$, assign the color $b_i$ to both $y_i a_{i1}$ and $y_i a_{i2}$. 
    For each $j$, assign the color $C_j$ to all edges incident to $c'_j$. 
    Finally, for each $i \in [s-1]$, assign the color $d_i$ to the edge $y_i y_{i+1}$. 
    The resulting edge-colored graph is star-colored. Indeed, for each $i$, the edges of color $b_i$ are $y_i a_{i1}$ and $y_i a_{i2}$, which form a star centered at $y_i$. For each $j$, the edges of color $C_j$ are those incident to $c'_j$, forming a star centered at $c'_j$. Finally, each color $d_i$ is assigned to a single edge, forming a trivial star. For an illustration, see Figure~\ref{fig:reduction}.

    Let $\alpha$ be a truth assignment for $\Phi$, and let $q$ denote the number of clauses satisfied by $\alpha$. 
    We construct a tree $T_\alpha$ as follows. 
    Start with the path $P = y_1y_2\cdots y_s$. 
    For each $i \in [s]$, add the edge $y_i a_{i1}$ if $\alpha(x_i)=1$, and add the edge $y_i a_{i2}$ if $\alpha(x_i)=0$. 
    For each satisfied clause $c_j$, pick exactly one true literal, and attach $c'_j a_{i1}$ if this literal is $x_i$, and $c'_j a_{i2}$ if it is $\neg x_i$.
    By construction, $T_\alpha$ is formed by iteratively attaching degree-one vertices to the initial path $P$, thus it is a tree. 
    All edge colors in $T_\alpha$ are distinct, so $T_\alpha$ is a rainbow tree. 
    Finally, $V(T_\alpha)$ consists of all $s$ vertices $y_i$, exactly one vertex from $\{a_{i1}, a_{i2}\}$ for each $i \in [s]$, and the $q$ clause vertices corresponding to the satisfied clauses. 
    Therefore, the order of the tree is $|V(T_\alpha)| = 2s+q$.

    Conversely, let $T$ be a rainbow tree in $G_\Phi$. 
    For each $i \in [s]$, the edges $y_i a_{i1}$ and $y_i a_{i2}$ share color $b_i$, so $T$ contains at most one of them, including at most one vertex from $\{a_{i1}, a_{i2}\}$.
    Similarly, for each $j$, $T$ contains at most one edge incident to $c'_j$, making $c'_j$ a leaf if it belongs to $T$. 
    Define a truth assignment $\alpha_T$ by setting $\alpha_T(x_i)=1$ if $a_{i1} \in V(T)$, $\alpha_T(x_i)=0$ if $a_{i2} \in V(T)$, and an arbitrary value otherwise.
    If $c'_j \in V(T)$, its unique neighbor is either $a_{i1}$ or $a_{i2}$, implying the corresponding literal satisfies $c_j$ under $\alpha_T$. 
    It follows that the number of clauses satisfied by $\alpha_T$ is at least the number of clause vertices in $T$.
    Observe that $|V(T) \cap \{y_1, \dots, y_s\}| \le s$ and, for the variable gadgets, $|V(T) \cap \{a_{i1}, a_{i2} \colon i \in [s]\}| \le s$. Therefore, at least $|V(T)| - 2s$ vertices of $T$ must be clause vertices.
    By the argument above, the corresponding truth assignment $\alpha_T$ satisfies at least $|V(T)| - 2s$ clauses in $\Phi$.
    
    Let $\opt$ denote the maximum number of satisfiable clauses in $\Phi$, and let $\opt'$ denote the maximum order of a rainbow tree in $G_\Phi$.
    The construction above shows that $\opt' \ge \opt + 2s$.
    Conversely, let $T$ be a rainbow tree of maximum order in $G_\Phi$. Then the corresponding assignment satisfies at least $|V(T)| - 2s = \opt' - 2s$ clauses, and hence $\opt \ge \opt' - 2s$.
    Therefore, $\opt' = \opt + 2s$.
    Since $G_\Phi$ can be constructed in polynomial time, this establishes that the maximum-size rainbow tree problem is NP-hard even on star-colored graphs.
\end{proof}

\begin{figure}[h]
    \centering
    \begin{tikzpicture}
        \tikzset{
            every circle node/.style={minimum size=0.3cm, inner sep=0cm},
            font=\small,
            tree edge/.style={line width=4pt, draw=#1, text=#1},
            normal edge/.style={line width=1.5pt, draw=#1!40, text=#1!80}
        }
        
        \def\dx{4.5}
        \def\dy{2.2}
        \def\cy{4.5}

        \node[circle, fill, label={270:$y_1$}] (y1) at (0, 0) {};
        \node[circle, fill, label={270:$y_2$}] (y2) at (\dx, 0) {};
        \node[circle, fill, label={270:$y_3$}] (y3) at (2*\dx, 0) {};
        
        \foreach \i/\pos in {1/y1, 2/y2, 3/y3} {
            \node[circle, fill, label={180:$a_{\i1}$}] (a\i1) at ($(\pos) + (-0.8, \dy)$) {};
            \node[circle, fill, label={0:$a_{\i2}$}]   (a\i2) at ($(\pos) + (0.8, \dy)$) {};
        }

        \node[circle, fill=blue,   minimum size=0.3cm, inner sep=0cm] at (a11) {};
        \node[circle, fill=teal,   minimum size=0.3cm, inner sep=0cm] at (a22) {};
        \node[circle, fill=purple, minimum size=0.3cm, inner sep=0cm] at (a31) {};

        \draw[tree edge=violet]          (y1) -- (y2) node[midway, below] {$d_1$};
        \draw[tree edge=cyan!80!black]   (y2) -- (y3) node[midway, below] {$d_2$};

        \draw[tree edge=blue]   (y1) -- (a11) node[midway, fill=white, inner sep=1pt] {$b_1$};
        \draw[normal edge=blue] (y1) -- (a12) node[midway, fill=white, inner sep=1pt] {$b_1$};

        \draw[normal edge=teal] (y2) -- (a21) node[midway, fill=white, inner sep=1pt] {$b_2$};
        \draw[tree edge=teal]   (y2) -- (a22) node[midway, fill=white, inner sep=1pt] {$b_2$};

        \draw[tree edge=purple]   (y3) -- (a31) node[midway, fill=white, inner sep=1pt] {$b_3$};
        \draw[normal edge=purple] (y3) -- (a32) node[midway, fill=white, inner sep=1pt] {$b_3$};

        \node[circle, fill, label={90:$c'_1$}] (c1) at (0, \cy) {};
        \node[circle, fill, label={90:$c'_2$}] (c2) at (0.666*\dx, \cy) {};
        \node[circle, fill, label={90:$c'_3$}] (c3) at (1.333*\dx, \cy) {};
        \node[circle, fill, label={90:$c'_4$}] (c4) at (2*\dx, \cy) {};

        \draw[tree edge=red]   (c1) -- (a11) node[pos=0.28, fill=white, inner sep=1pt] {$C_1$};
        \draw[normal edge=red] (c1) -- (a21) node[pos=0.3, fill=white, inner sep=1pt] {$C_1$};

        \draw[normal edge=orange] (c2) -- (a11) node[pos=0.28, fill=white, inner sep=1pt] {$C_2$};
        \draw[tree edge=orange]   (c2) -- (a22) node[pos=0.3, fill=white, inner sep=1pt] {$C_2$};

        \draw[normal edge=yellow!80!black] (c3) -- (a12) node[pos=0.30, fill=white, inner sep=1pt] {$C_3$};
        \draw[tree edge=yellow!80!black]   (c3) -- (a31) node[pos=0.32, fill=white, inner sep=1pt] {$C_3$};

        \draw[normal edge=green!70!black] (c4) -- (a12) node[pos=0.2, right, fill=white, inner sep=1pt] {$C_4$};
        \draw[normal edge=green!70!black] (c4) -- (a32) node[pos=0.32, fill=white, inner sep=1pt] {$C_4$};
    \end{tikzpicture}
    \caption{Illustration of the reduction from MAX-SAT to the Maximum-Size Rainbow Tree problem in star-colored graphs. The figure shows the edge-colored graph $G_\Phi$ associated with the formula $\Phi=(x \lor y)\land(x \lor \neg y)\land(\neg x \lor z)\land(\neg x \lor \neg z)$.
    The thick colored edges form a maximum-size rainbow tree corresponding to the optimal assignment $(x,y,z)=(1,0,1)$.}
    \label{fig:reduction}
\end{figure}
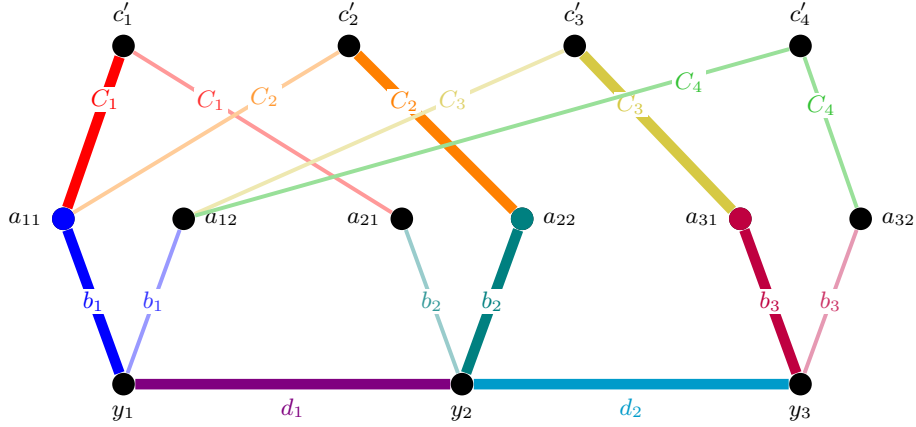

\section{Algorithms}
\label{sec:algo}

In this section, we provide a polynomial-time algorithm that finds a properly colored tree of order at least $\min\{|V(G)|, 2\delta^c(G)+1\}$ in a connected graph $G$, provided such a tree exists.
We briefly describe the main ideas of the algorithm.

First, we preprocess the input graph by deleting every edge whose removal does not decrease the color degree of either endpoint.
If such a deletion creates a cut edge, and the color degrees in both resulting components remain unchanged, then the cut edge can be used to explicitly construct a properly colored tree of sufficiently large size by combining it with edges incident to its endpoints. 
From this certificate, we construct a properly colored tree on at least $2\delta^c(G)+1$ vertices.
Otherwise, we obtain a reduced graph $G^\star$.

Next, we reduce the problem to finding a rainbow tree in a connected star-colored graph. 
We recolor each monochromatic component of $G^\star$ with a new color.
This recoloring produces a connected star-colored graph $G'$ and preserves the minimum color degree.
Moreover, every rainbow tree in $G'$ corresponds to a properly colored tree in $G$.

To construct a large rainbow tree in a connected star-colored graph, the algorithm maintains a rainbow tree $T$ on a vertex set $N$ through greedy extensions.
If no greedy extension using a new color is possible, the algorithm performs one- or two-edge exchanges within $G[N]$.
These exchanges are found via matroid intersection on the restricted graphic and color partition matroids; this procedure guarantees extendability and updates $T$ while preserving the rainbow property.

First, in Section~\ref{sec:star}, we present an algorithm for finding a large rainbow tree in a connected star-colored graph.
Building on this result, Section~\ref{sec:general} establishes our main theorem for general connected graphs.

\subsection{Star-colored graphs}
\label{sec:star}

Although the problem of finding a maximum-size rainbow tree is NP-hard for star-colored graphs, we can find a rainbow tree of order at least $\min\{|V(G)|, 2\delta^c(G)+1\}$ in polynomial time.
Our algorithm is presented as Algorithm~\ref{algo:rainbow-tree}.

\begin{algorithm}[h]
\caption{\textsf{BuildRainbowTree}$(G)$}\label{algo:rainbow-tree}
\DontPrintSemicolon

\KwIn{A connected star-colored graph $G=(V,E)$ with edge coloring $c\colon E\to[k]$.}
\KwOut{A rainbow tree $T$ in $G$.}

\medskip

$\delta^c \gets \delta^c(G)$.\;
Pick an arbitrary vertex $r \in V$.\;
Initialize $T \gets (\{r\}, \emptyset)$, $N \gets \{r\}$, and $U \gets \emptyset$.\;

Let $\cM_1$ be the graphic matroid of $G$, and let $\cM_2$ be the partition matroid on $E(G)$ with rank $1$ for each color class.\;

\While{$|N| < \min\{|V|, 2\delta^c+1\}$\label{step_rt:while}}{
    \eIf{there exists an edge $e = ux \in \partial(N)$ with $u\in N$ and $c(e) \notin U$\label{step:single-edge-s}}{
        Update $T \gets (N + x, E(T) + e)$.\;
        $N \gets V(T)$, $U \gets c(E(T))$.\label{step:single-edge}\;
    }{
        \For{each edge $e=ux \in \partial(N)$ with $u\in N$\label{step:exchange-one-edge-s}{\footnotesize\color{Blue}{\tcp*{One-edge exchange}}}}{
            Find a maximum common independent set $F$ of $\cM_1|E(G[N])$ and $(\cM_2/e)|E(G[N])$.\;
            \If{$|F| = |N| - 1$}{
                Update $T \gets (N + x, F + e)$.\;
                $N \gets V(T)$, $U \gets c(E(T))$.\;
                \textbf{goto} Step \ref{step_rt:while}.\label{step:exchange-one-edge}\;
            }
        }
        \For{each pair of edges $e_0=ux \in \partial(N)$ and $e = xy$ such that $u\in N$, $x,y \notin N$, and $c(e_0)\neq c(e)${\footnotesize\color{Blue}{\tcp*{Two-edge exchange}}}}{
            \For{each vertex $v \in N \setminus \{u\}$}{
                $N' \gets N - v$.\;
                Find a maximum common independent set $F$ of $\cM_1|E(G[N'])$ and $(\cM_2/\{e_0,e\})|E(G[N'])$.\;
                \If{$|F| = |N'| - 1$}{
                    Update $T \gets (N' \cup \{x,y\}, F \cup \{e_0, e\})$.\;
                    $N \gets V(T)$, $U \gets c(E(T))$.\;
                    \textbf{goto} Step \ref{step_rt:while}\label{step:exchange-two-edge}.\;
                }
            }
        }
        \textbf{break}.\;
    }
}
\Return{$T$}\;
\end{algorithm}

\begin{thm}\label{thm:rb-star}
    Given a connected star-colored graph $G$, Algorithm \ref{algo:rainbow-tree} finds a rainbow tree of order at least $\min\{|V(G)|, 2\delta^c(G)+1\}$ in polynomial time, if one exists.
\end{thm}

\begin{proof}
    We assume that $G$ has a rainbow tree of order at least $\min\{|V(G)|, 2\delta^c(G)+1\}$.
    By Theorem~\ref{thm:Rb-ex}, we have $G \notin \mathcal{G}^1 \cup \mathcal{G}^2 \cup \mathcal{G}^4 \cup \mathcal{G}^6$.
    Let $T$ be the tree returned by Algorithm~\ref{algo:rainbow-tree}, and let $U = c(E(T))$. 
    Define the edge set $\hat{E} = E(G[V(T)]) \cup \partial(V(T))$. 
    Let $\cM_1$ be the graphic matroid of $(V(G), \hat{E})$, and let $\cM_2$ be the partition matroid on $\hat{E}$ in which each color class has rank $1$.
    It follows that $E(T) \in \cI(\cM_1) \cap \cI(\cM_2)$.

    Let $D_I$ denote the exchange digraph of $E(T)$ with respect to $\cM_1$ and $\cM_2$. 
    Specifically, $V(D_I) = \hat{E}$. For $e \in E(T)$ and $f \in \hat{E} \setminus E(T)$, the arc set of $D_I$ contains $(e, f)$ if $E(T) - e + f \in \cI(\cM_1)$, and it contains $(f, e)$ if $E(T) - e + f \in \cI(\cM_2)$.
    For $i \in \{1,2\}$, define $Z_i = \{f \in \hat{E} \setminus E(T)\colon E(T) + f \in \cI(\cM_i)\}$. 
    Observe that $Z_1 = \partial(V(T))$.

    Call an edge $f \in \hat{E}$ \emph{reachable} if $D_I$ contains a directed path from $Z_1$ to $f$.
    Since $T$ is rainbow, each color $c \in U$ appears on a unique edge $e_c \in E(T)$. 
    We say that a color $c \in U$ is \emph{reachable} if its corresponding edge $e_c$ is reachable. 
    Let $R = \{c_1, \dots, c_r\}$ be the set of all reachable colors.

    \begin{cl}\label{cl:rbmt}
        The graph $G[V(T)]-(E_{c_1} \cup \cdots \cup E_{c_r})$ has exactly $r+1$ components, and the edges of $\partial(V(T))$ are contained in $E_{c_1} \cup \cdots \cup E_{c_r}$.
    \end{cl}
    
    \begin{proof}
        Since Steps~\ref{step:single-edge-s}-\ref{step:single-edge} and \ref{step:exchange-one-edge-s}-\ref{step:exchange-one-edge} of Algorithm~\ref{algo:rainbow-tree} yield no augmentation, the standard matroid intersection algorithm implies that $D_I$ contains no directed path from $Z_1$ to $Z_2$.

        Consider an edge $f \in Z_1 = \partial(V(T))$.
        If $c(f) \notin U$, then $E(T) + f \in \cI(\cM_2)$, which implies $f \in Z_2$.
        This inclusion yields $Z_1 \cap Z_2 \neq \emptyset$, contradicting the absence of a directed path from $Z_1$ to $Z_2$.
        Thus, $c(f) \in U$. 
        Let $c = c(f)$. 
        By the definition of $\cM_2$, the set $E(T) - e_c + f \in \cI(\cM_2)$, which implies that $(f, e_c)$ is an arc in $D_I$. 
        Because $f \in Z_1$ is trivially reachable, $e_c$ is also reachable, yielding $c \in R$. 
        Therefore, $\partial(V(T)) \subseteq \bigcup_{i=1}^r E_{c_i}$.
    
        Observe that removing the edges $\{e_c\colon c \in R\}$ from $T$ produces a forest with exactly $r+1$ components.
        Because $E(T) \setminus \{e_c\colon c \in R\} \subseteq E(G[V(T)]) \setminus \bigcup_{i=1}^r E_{c_i}$, the graph $G[V(T)] \setminus \bigcup_{i=1}^r E_{c_i}$ has at most $r+1$ components.

        It remains to show that $G[V(T)] - (E_{c_1} \cup \cdots \cup E_{c_r})$ has at least $r+1$ components.
        Suppose, for a contradiction, that there exists an edge $f = xy \in E(G[V(T)])$ with endpoints $x$ and $y$ in different components of $T \setminus \bigcup_{i=1}^r E_{c_i}$ and such that $c(f) \notin R$.
        Let $P$ be the unique $x$-$y$ path in $T$; then $P$ contains at least one edge of $F = {e_{c_1}, \dots, e_{c_r}}$. Let $e$ be the first such edge along $P$ from $x$ to $y$. In the forest $T \setminus e$, the vertices $x$ and $y$ lie in different components. Since $f = xy$, it connects these components, so $E(T) - e + f \in \mathcal{I}(\cM_1)$, implying that $(e,f)$ is an arc in $D_I$. Because $e \in F$, we have $c(e) \in R$, and by definition $e$ is reachable. Therefore, $f$ is also reachable. If $c(f) \notin U$, then $E(T) + f \in \mathcal{I}(\cM_2)$, yielding $f \in Z_2$. Since $f$ is reachable from $Z_1$, this produces a directed path from $Z_1$ to $Z_2$ in $D_I$, a contradiction. If $c(f) \in U$, then $E(T) - e_{c(f)} + f \in \mathcal{I}(\mathcal{M}2)$, so $(f, e_{c(f)})$ is an arc in $D_I$. Since $f$ is reachable, $e_{c(f)}$ is reachable, implying $c(f) \in R$, again a contradiction. Hence, no such edge $f$ exists, and $G[V(T)] \setminus \bigcup_{i=1}^r E_{c_i}$ has exactly $r+1$ components, completing the proof of Claim~\ref{cl:rbmt}.
    \end{proof}

    \begin{cl}\label{cl:small-n-span}
        If Algorithm~\ref{algo:rainbow-tree} terminates with $|V(T)| < |V(G)|$, then $G$ does not contain a rainbow spanning tree.
    \end{cl}
    
    \begin{proof}
        Suppose, for contradiction, that $G$ has a rainbow spanning tree.
        Let $H \coloneqq G \setminus \bigcup_{i=1}^r E_{c_i}$, the graph obtained by removing all edges of the colors used in $T$.        
        Since every edge between $V(T)$ and $V(G)\setminus V(T)$ is in
        $E_{c_1}\cup\cdots\cup E_{c_r}$, $H$ contains no edges between $V(T)$ and $V(G) \setminus V(T)$.
        Since $V(G)\setminus V(T)$ is nonempty, this set of vertices forms at least one connected component in $H$ apart from the components inside $V(T)$.
        By Claim~\ref{cl:rbmt}, $H[V(T)]$ has exactly $r+1$ components, so $H$ has at least $r+2$ components in total.
        However, by Theorem~\ref{thm:Rb}, removing $r$ colors from a graph that contains a rainbow spanning tree cannot leave more than $r+1$ components. Since $r \le |E(T)| = |V(T)|-1 \le |V(G)|-2$, we get a contradiction. Hence, $G$ cannot have a rainbow spanning tree.
    \end{proof}
    
    If $2\delta^c(G)+1 \ge |V(G)|$, then $G$ admits a rainbow spanning tree, and
    by Claim~\ref{cl:small-n-span}, the algorithm outputs one.
    Therefore, for the rest of the proof, we assume that $|V(G)|>2\delta^c(G)+1$.
    Let $W_1,  \dots, W_{r+1}$ be the components of $G[V(T)] \setminus \bigcup_{i=1}^r E_{c_i}$.
    For every monochromatic star isomorphic to $K_2$, we choose one of its end vertices arbitrarily as its center.
    For any $v \in V(G)$ and any color $c$, if $v$ is the center of a monochromatic star of color $c$, then we say $c$ \emph{belongs} to $v$.
    For each $v \in V(G)$, let
    $$
    C(v) = \{c \in R\colon c \text{ belongs to } v\},
    $$
    and define
    $$
    S = \{v \in V(G)\colon C(v) \neq \emptyset\}.
    $$
    Thus, $S$ is precisely the set of centers of monochromatic stars whose colors lie in $R$.
    
    \begin{cl}\label{cl:center}
        The set $S$ is contained in $V(T)$.
    \end{cl}
    \begin{proof}
         Let $v \in S$. Then some color $c \in R$ belongs to $v$, so $v$ is the center of a monochromatic star of color $c$.
         Since every color in $R$ appears on an edge of $T$, and the center of such a star is an endpoint of one of its edges, we have $v \in V(T)$.
    \end{proof}

    Reindex $W_1,\dots,W_{r+1}$ so that there exists an integer $1 \le m \le r+1$ such that $V(W_i)\cap S \neq \emptyset$ for $i\in[m]$, and $V(W_i)\cap S = \emptyset$ for $m+1 \le i \le r+1$.
    For each $i\in[m]$, fix a vertex $x_i \in V(W_i)\cap S$. 
    Define
    $\beta_i\coloneqq |(S\setminus\{x_1,\dots,x_m\})\cap V(W_i)|$ and set $\beta\coloneqq \sum_{i=1}^m \beta_i$.
    
    \begin{cl}\label{cl:2delta}
    The tree $T$ has order $|V(T)| \ge 2\delta^c(G)$.
    \end{cl}

    \begin{proof}
        We first show that
        \begin{equation}\label{eq:size-lb}
        |V(T)| \ge m + 2\beta + 1.
        \end{equation}        
        Indeed, for $i\in [m]$, we have $|W_i|\ge 1+\beta_i$. 
        Also, $|W_j|\ge 1$ for $m+1\le j\le r+1$. Thus, 
        \begin{equation}\label{eq:size-lb-mid}
        |V(T)|=\sum_{j=1}^{r+1}|W_j| \ge \sum_{i=1}^m(1+\beta_i) + (r+1-m) = \beta+r+1.
        \end{equation}
        Since each vertex in $S$ supports at least one color in $R$, we have $r\ge |S|=m+\beta$, and thus $|V(T)|\ge m+2\beta+1$, proving \eqref{eq:size-lb}.
        
        By Claim~\ref{cl:rbmt}, all colors on edges between $x_i$ and $V(G) \setminus V(T)$ are in $C(x_i)$.
        In $W_i$, the number of colors on edges incident to $x_i$ that are not in $C(x_i)$ is at most $|W_i|-1$.
        Let $v \in V(T) \setminus V(W_i)$.
        If $c(x_i v) \notin C(x_i)$, then $v \in S \setminus V(W_i)$.
        The set $S \setminus V(W_i)$ contains exactly the vertices $x_j$ for $j \neq i$, and the other $\sum_{j \neq i} \beta_j$ vertices of $S \setminus \{x_1, \dots, x_m\}$.
        Therefore, 
        \begin{equation}\label{eq:deg-bound-single}
        d_G^c(x_i) \le |C(x_i)| + (|W_i|-1)
        + \bigl|\{c(x_ix_j)\colon j\neq i,\ c(x_ix_j)\notin C(x_i)\}\bigr|
        + \sum_{j\neq i}\beta_j.
        \end{equation}
        Summing \eqref{eq:deg-bound-single} over $1\le i\le m$ gives
        \begin{equation}\label{eq:deg-bound-sum}
        \sum_{i=1}^m d_G^c(x_i) \le \sum_{i=1}^m(|C(x_i)|+|W_i|) + (m-1)\beta
        + \sum_{i=1}^m \bigl|\{c(x_ix_j)\colon j\neq i,\ c(x_ix_j)\notin C(x_i)\}\bigr| - m.
        \end{equation}

        To bound the right hand side, we orient each edge of $G[\{x_1,\dots,x_m\}]$ as follows.
        For each edge $x_ix_j \in E(G)$ with $i \neq j$, orient it from $x_i$ to $x_j$ if $c(x_ix_j) \in C(x_i)$; otherwise, orient it from $x_j$ to $x_i$.
        Then, for each $i$, the set $\{c(x_ix_j)\colon j\neq i,\ c(x_ix_j)\notin C(x_i)\}$ is counted by the indegree of $x_i$ in this orientation.
        Hence,
        $$
        \sum_{i=1}^m \bigl|\{c(x_ix_j)\colon j \neq i,\ c(x_ix_j) \notin C(x_i)\}\bigr| \leq |E(G[\{x_1,\dots,x_m\}])| \leq \frac{m(m-1)}{2},
        $$
        and equality holds only if $G[\{x_1,\dots,x_m\}]$ is a complete graph.
        Also, since every vertex of $S \setminus\{x_1,\dots,x_m\}$ supports at least one color in $R$, we have
        $\sum_{i=1}^m |C(x_i)| \le r-\beta$.
        Finally, for $m+1\le j\le r+1$, we have $|W_j|\ge 1$.
        Thus, $\sum_{i=1}^m |W_i| \le |V(T)|-(r+1-m)=|V(T)|-r+m-1$.
        Plugging these bounds into \eqref{eq:deg-bound-sum} yields
        \begin{equation}\label{eq:deg-bound-final}
        m\delta^c(G) \le \sum_{i=1}^m d_G^c(x_i)
        \le |V(T)| + (m-2)\beta + \frac{m(m-1)}{2}-1.
        \end{equation}

        Suppose indirectly that $|V(T)| \leq 2 \delta^{c}(G)-1$.
        Then, by \eqref{eq:deg-bound-final}, we obtain
        $$
        (m-2)2 \delta^{c}(G) \leq(m-2)(m+2 \beta+1)-2 .
        $$
        On the other hand, by \eqref{eq:size-lb}, we have 
        $$
        2 \delta^{c}(G) \geq m+2 \beta+2.
        $$
        Since $m\geq 1$, these together imply $m=1$. Choose $w_1 \in V\left(W_1\right)$ and $w_2 \in V(W_2)$ such that $w_1 \in S$ and $w_1 w_2 \in E(T)$.
        Note that $w_2 \notin S$.
        Since $m=1$, we have 
        $$
        \begin{aligned}
        d_G^{c}\left(w_1\right)+d_G^{c}\left(w_2\right) & \leq\left(\left|W_1\right|-1+\left|C\left(w_1\right)\right|\right)+\left(\left|W_2\right|-1+r-\left|C\left(w_1\right)\right|+1\right) \\
        & \leq\left|W_1\right|+\left|W_2\right|+r-1 \\
        & \leq|T| \\
        & \leq 2 \delta^{c}(G)-1.
        \end{aligned}
        $$
        By the definition of minimum color degree, we have $d_G^c(w_1) + d_G^c(w_2) \geq 2\delta^c(G)$, which is a contradiction.
        Therefore, $|V(T)|\ge 2\delta^c$.
        This completes the proof of Claim~\ref{cl:2delta}.
    \end{proof}
    
    By Claim~\ref{cl:2delta}, the algorithm outputs a tree of order at least $2\delta^c(G)$. To finish the proof of the theorem, we show that if $|V(T)|=2\delta^{c}(G)$, then $G\in\cG^1\cup\cG^2\cup\cG^4\cup\cG^6$, which contradicts our initial assumption. 
    Suppose that $|V(T)|=2 \delta^{c}(G) \neq|V(G)|$.
    By \eqref{eq:deg-bound-final}, we have
    $$
    2 \delta^{c}(G)(m-2) \leq(m+2 \beta+1)(m-2).
    $$
    Furthermore, \eqref{eq:size-lb} yields
    $$
    2 \delta^{c}(G) \geq m+2 \beta+1.
    $$
    Therefore, if $m \geq 3$, then $2\delta^c(G) = m+2\beta+1$ and the inequalities \eqref{eq:size-lb}--\eqref{eq:deg-bound-final} hold with equality.
    If $m = 2$, then the condition $|V(T)| = 2\delta^c(G)$ ensures that equality holds in \eqref{eq:deg-bound-single}, \eqref{eq:deg-bound-sum}, and \eqref{eq:deg-bound-final}.
    Thus, for $m \geq 2$, we have $d^c_G(x_i) = \delta^c(G)$ for each $i\in[m]$. 
    In particular, $|E(G[\{x_1, \dots, x_m\}])| = \frac{m(m-1)}{2}$, meaning that $G[\{x_1, \dots, x_m\}]$ is a complete graph.

    We distinguish two cases based on whether $m\geq 2$ or $m=1$.

    \begin{case}
        $m \ge 2$.
    \end{case}

    In this case, we show that $G \in \cG^1 \cup \cG^2 \cup \cG^4$.

    \begin{cl}\label{cl:adjacency-tight}
    For each $i\in[m]$, the vertex $x_i$ is adjacent to every vertex in $(V(W_i)\cup S)\setminus \{x_1,\dots,x_m\}$. Moreover, the edges from $x_i$ to these vertices receive pairwise distinct colors, and none of these colors belong to $C(x_i)$. Furthermore, if $|V(W_i)\cap S| \ge 2$, then for every $x \in V(W_i)\cap S$ we have $|C(x)|=1$, and for every $j \neq i$ the color of the edge $xx_j$ lies in $C(x)$.
    \end{cl}

    \begin{proof}
    The first statement follows from equality in \eqref{eq:deg-bound-single}.
    Equality can hold only if $x_i$ is adjacent to all vertices in $V(W_i) \setminus \{x_i\}$ and to all vertices in $S \setminus \{x_1, \dots, x_m\}$, and these edges have distinct colors.
    Furthermore, these colors do not belong to $C(x_i)$.

    For the second statement, since $\sum_{i=1}^m |C(x_i)| = r-\beta$, we have $|C(x)| = 1$ for each $x \in V(W_i) \cap S \setminus \{x_i\}$.
    Since the choice of $x_i \in V(W_i) \cap S$ is arbitrary, it follows that $|C(x)| = 1$ for all $x \in V(W_i) \cap S$.
    By the first statement, for each $j \neq i$, we have $xx_j \in E(G)$ and $c(xx_j) \in C(x)$.
    This completes the proof of Claim~\ref{cl:adjacency-tight}.
    \end{proof}

    \begin{subcase}
        $m \ge 3$.
    \end{subcase}

    By the equality in \eqref{eq:size-lb-mid}, we have $|V(T)|=\beta+r+1$.
    Since $|V(T)|=m+2\beta+1$, we have $r=m+\beta=|S|$.
    Hence, every vertex in $S$ is the center of exactly one color in $R$.
    In particular, $|C(v)|=1$ for all $v \in S$.
    By the equality in \eqref{eq:size-lb-mid}, we have $\sum_{i=1}^m |W_i| = m+\beta$ and $|W_{m+1}|=\cdots=|W_{r+1}|=1$.
    For $i\in[m]$, we have $|W_i| \geq |V(W_i)\cap S| = 1+\beta_i$.
    Furthermore, we have $\sum_{i=1}^m(1+\beta_i)=m+\beta$.
    Thus, $|W_i|=|V(W_i)\cap S|$ for each $i\in[m]$.

    Next, we show that each $W_i$ has exactly one vertex.
    
    \begin{cl}\label{cl:Wi-singleton}
        For each $i\in[m]$, we have $V(W_i)=\{x_i\}$.
        Consequently, $\beta=0$, $r=m$, and $S=\{x_1,\dots,x_m\}$.
    \end{cl}

    \begin{proof}
        Assume that $|W_i| \geq 2$ for some $i \in [m]$.
        Choose $x_i' \in V(W_i) \setminus \{x_i\}$.
        By Claim~\ref{cl:adjacency-tight}, we have $c(x_i' x_j) \in C(x_i')$ for each $j \neq i$.
        Hence,
        $$
        \bigl|\{c(x_i' x_j)\colon j \neq i \text{ and } c(x_i' x_j) \notin C(x_i')\}\bigr| = 0.
        $$
        Now fix $i\in[m]$.
        We have $|W_i| = 1 + \beta_i$ and $|C(x_i)|=1$.
        By the equality in \eqref{eq:deg-bound-single}, we have
        $$
        \delta^{c}(G) = d_G^{c}(x_i) = 1 + \beta + |\{c(x_i x_j)\colon j \neq i \text{ and } c(x_i x_j) \notin C(x_i)\}|.
        $$
        Thus, for any $1 \leq i < i' \leq m$, we have
        $$
        |\{c(x_i x_j)\colon j \neq i \text{ and } c(x_i x_j) \notin C(x_i)\}| = |\{c(x_{i'}x_j)\colon j \neq i' \text{ and } c(x_{i'} x_j) \notin C(x_{i'})\}|.
        $$
        Since $G[\{x_1,\dots,x_m\}]$ is a complete graph, we have
        \[
        \bigl|\{c(x_i x_j)\colon j \neq i \text{ and } c(x_i x_j) \notin C(x_i)\}\bigr| = \frac{m-1}{2}
        \]
        for each $i\in[m]$.
        The vertices $x_i$ and $x_i'$ play symmetric roles in $W_i$.
        Therefore, we can replace $x_i$ by $x_i'$.
        Since $m \geq 3$, we have $\frac{m-1}{2} > 0$.
        This is a contradiction.
        Hence, $|W_i|=1$ for all $i\in[m]$.
        Then $\beta_i=0$ for all $i\in[m]$.
        Thus, $\beta=0$ and $r=m$.
        Since $|S|=r$, we have $S=\{x_1,\dots,x_m\}$.
        This completes the proof of Claim~\ref{cl:Wi-singleton}.
    \end{proof}

    Since $r=m$, exactly $r+1-m=1$ component is disjoint from $S$, and this component is a singleton.
    Therefore, $V(T) = S \cup \{z\}$ for a unique vertex $z \in V(T) \setminus S$.
    Define an orientation $D$ of $G[S]$ by directing the edge $x_i x_j$ from $x_i$ to $x_j$ whenever $c(x_i x_j) \in C(x_i)$.
    Then $D$ is a regular tournament in which each vertex has an in-degree of $(m-1)/2$, and hence $m$ is odd.

    \begin{cl}\label{cl:strong-connected}
        The tournament $D$ is strongly connected.
    \end{cl}
    
    \begin{proof}
        Suppose for a contradiction that $D$ is not strongly connected.
        Then there is a partition of $V(D)$ into two nonempty sets $A$ and $B$ such that all arcs between $A$ and $B$ are directed from $A$ to $B$.
        Let $a = |A|$ and $b = |B|$.
        Then $a+b=m$ and $a,b \geq 1$.
        By summing the out-degrees of the vertices in $A$, we obtain
        \[
        \frac{a(m-1)}{2} = \sum_{v \in A} d^+(v) = \frac{a(a-1)}{2} + ab.
        \]
        Since $a \geq 1$, dividing by $a/2$ yields $m-1 = a-1+2b$.
        Thus, $m = a+2b$.
        Recall that $a+b=m$.
        Then we have $b=0$, which contradicts $b \geq 1$.
        Hence, $D$ is strongly connected.
    \end{proof}

    \begin{cl}\label{cl:core-tree-avoid}
        For each $x_i \in S$, the complete graph $G[S]$ admits a rainbow spanning tree $T_i$ that contains no colors in $C(x_i)$.
    \end{cl}
    
    \begin{proof}
        By Claim~\ref{cl:strong-connected}, $D$ is strongly connected.
        Hence, $D$ contains an in-arborescence $A_i$ rooted at $x_i$.
        Let $T_i$ be the spanning tree of $G[S]$ whose edge set is the set of underlying undirected edges of $A_i$.
        By the definition of an in-arborescence, every non-root vertex $x_j$ has exactly one outgoing arc in $A_i$.
        Therefore, $T_i$ contains exactly one edge with color in $C(x_j)$.
        Since $x_i$ is the root of $A_i$, no edge of $T_i$ has color in $C(x_i)$.
        Recall that $|C(x_j)| = 1$ for all $x_j \in S$.
        Therefore, $T_i$ is a rainbow spanning tree of $G[S]$, and contains no color in $C(x_i)$.
    \end{proof}
    
    \begin{cl}\label{cl:indep}
        The set $V(G) \setminus S$ is an independent set in $G$.
    \end{cl}

    \begin{proof}
        Note that there is no edge between $z$ and $V(G) \setminus V(T)$, it suffices to show that $V(G) \setminus V(T)$ is independent.
        Assume for a contradiction that $V(G) \setminus V(T)$ is not an independent set.
        Since $G$ is connected, there exist adjacent vertices $x, y \in V(G) \setminus V(T)$ such that $x$ is adjacent to a vertex in $V(T)$.
        Let $x_i \in V(T)$ be a neighbor of $x$.
        Since every monochromatic subgraph of $G$ is a star, Claim~\ref{cl:rbmt} implies that $c(xy) \neq c(xx_i)$ and $c(xy) \notin \{c_1, c_2, \dots, c_r\}$.
        By Claim~\ref{cl:core-tree-avoid}, the graph $G[\{x_1, x_2, \dots, x_m\}]$ contains a rainbow spanning tree $T^{\prime}$ that avoids the color $c(xx_i)$.
        Hence, $T^{\prime} + xx_i + xy$ is a rainbow tree.
        This contradicts Steps~\ref{step:exchange-one-edge}--\ref{step:exchange-two-edge} of Algorithm~\ref{algo:rainbow-tree}, because the algorithm would not have terminated.
        Thus, $V(G) \setminus V(T)$ is independent.
        This completes the proof of Claim~\ref{cl:indep}.
    \end{proof}

    Let $v \in V(G) \setminus S$.
    For each neighbor $x_i \in S$ of $v$, the edge $vx_i$ has the unique color in $C(x_i)$, and $v$ has no neighbors outside $S$.
    Thus, we have
    $d^c_G(v) = |N_G(v) \cap S|$.
    Recall that $\delta^c(G) = d^c_G(x_i) = \frac{m+1}{2}$ for each $x_i \in S$.
    Hence, every vertex $v \in V(G) \setminus S$ satisfies $|N_G(v) \cap S| \geq \delta^c(G) = \frac{m+1}{2}$.
    Therefore, $G$ is a split graph with clique $S$ and independent set $V(G) \setminus S$.
    Each vertex in the independent set has at least $\frac{m+1}{2}$ neighbors in $S$.
    Furthermore, the edge-coloring of $G[S]$ is induced by the tournament $D$, in the sense that the color of each edge is determined by its tail in $D$.
    Thus, $G$ belongs to the class $\cG^1$ when $m \ge 3$.

    \begin{subcase}
        $m=2$.
    \end{subcase}

    The detailed proof is deferred to Appendix~\ref{app:m1m2}, since it closely follows the argument of Hu et al.~\cite{hu2022maximum}, with only minor modifications for our setting.
    It shows that if $m=2$, then either $V(W_2)\subseteq S$, in which case $G\in \cG^2$,
    or $V(W_2)\not\subseteq S$, in which case $G\in \cG^4$.

    \begin{case}
        $m=1$.
    \end{case}

    The detailed proof is deferred to Appendix~\ref{app:m1m2}, since it closely follows the argument of Hu et al.~\cite{hu2022maximum}, with only minor modifications for our setting.
    It shows that if $m=1$, then $G\in \cG^6$.
    \medskip
    
    Finally, let us elaborate on the time complexity of the algorithm. The while-loop iterates at most $\min\{|V(G)|, 2\delta^c(G)+1\} \leq |V(G)|$ times.
    Each iteration performs at most $O(|V(T)| \cdot |E(T)|^2)$ calls to the matroid intersection algorithm.
    Each call can be solved in polynomial time.
    Therefore, Algorithm~\ref{algo:rainbow-tree} runs in polynomial time. This completes the proof of Theorem~\ref{thm:rb-star}.
\end{proof}

\subsection{General case}
\label{sec:general}

This section is devoted to the proof of our main result.
The algorithm is presented as Algorithm~\ref{algo:properly-tree}.

\begin{algorithm}[h]
\caption{\textsf{BuildProperlyColoredTree}$(G)$}\label{algo:properly-tree}
\DontPrintSemicolon

\KwIn{A connected edge-colored graph $G=(V,E)$ with edge coloring $c\colon E\to[k]$.}
\KwOut{A properly colored tree $T$ with $|V(T)|\ge \min\{|V(G)|,2\delta^c(G)+1\}$, if one exists, otherwise \textsf{NO}.}

\medskip

Let $\delta_0\geq 3$ be any fixed constant.\;
$\delta^c \gets \delta^c(G)$.\;

\If{$\delta^c \le \delta_0$\label{step:constant}}{
    \eIf{an exhaustive search finds a properly colored tree $T$ of order $\min\{|V(G)|,2\delta^c+1\}$ in $G$}{
            \Return{$T$}.\;
        }{
            \Return{\textsf{NO}}.\;
        }
}

\lIf{$G \in \cG^1 \cup \cG^2 \cup \cG^3 \cup \cG^5$}{\Return{\textsf{NO}}.\label{step:extremal}}

$G^{\star} \gets G$.\;
\While{there exists an edge $e=vw \in E(G^{\star})$ such that $d^{c}_{G^{\star}-e}(v)=d^{c}_{G^{\star}}(v)$ and $d^{c}_{G^{\star}-e}(w)=d^{c}_{G^{\star}}(w)$\label{step:while-s}}{
    \eIf{$e$ is a cut edge in $G^{\star}$}{
        Construct a properly colored tree $T$ of order $2\delta^c+1$ (as per Lemma \ref{lem:bridge-certificate}).\;
        \Return{$T$}.\label{step:cut-end}\;
    }{
        $G^{\star} \gets G^{\star}-e$.\label{step:delete}\;
    }
}
$G' \gets G^{\star}$.\;
Construct an edge-colored graph $\bar{G}'$ from $G'$ by recoloring each monochromatic component $S$ of $G'$ with a distinct new color $c_S$.\label{step:recolor}\;

\eIf{$\bar{G}' \in \cG^1 \cup \cG^2 \cup \cG^4 \cup \cG^6$}{
    Choose an edge $e_0$ deleted in Step~\ref{step:delete} and an inclusion-minimal edge set $E^\bullet$ with $e_0 \notin E^\bullet$ such that each monochromatic component of $G'' = G' \cup \{e_0\} \setminus E^\bullet$ is a star and $\delta^c(G'') = \delta^c$.\;
    Reconstruct $\bar{G}''$ from $G''$ by recoloring its monochromatic components as in Step~\ref{step:recolor}.\label{step:recolor-again}\;
    $T \gets \textsf{BuildRainbowTree}(\bar{G}'')$.\;
}{
    $T \gets \textsf{BuildRainbowTree}(\bar{G}')$.\;
}
Restore the original colors of $E(T)$ as in $G$.\;
\Return{$T$}.\;
\end{algorithm}

\begin{thm}\label{thm:main}
    Given a connected edge-colored graph $G$, Algorithm~\ref{algo:properly-tree} finds a properly colored tree of order at least $\min\{|V(G)|, 2\delta^c(G)+1\}$ in polynomial time, if one exists.
\end{thm}

\begin{proof}
    Let $\delta^c=\delta^c(G)$. 
    If $\delta^c \le \delta_0$, then the algorithm checks all trees on at least $2\delta^c+1$ vertices and returns one if such a tree exists.
    
    By Theorem~\ref{thm:PC-ex}, if $G\in\cG^1\cup\cG^2\cup\cG^3\cup\cG^5$, then $G$ has no properly colored tree of order at least $\min\{|V(G)|, 2\delta^c+1\}$, and hence the algorithm correctly returns NO in Step~\ref{step:extremal}.
    If $G\in \cG^2$, then $\delta^c(G)=2$, a case resolved in Step~\ref{step:constant}. 
    Otherwise, we have $\delta^c(G)\neq 2$, and $\delta^c(G')=\delta^c(G)$ yields $\delta^c(G')\neq 2$, hence we may assume that $G\notin \cG^2$ and $\bar{G}' \notin \cG^2$.
    
    \begin{lem}\label{lem:bridge-certificate}
    For any graph $G_0$ with $\delta^c(G_0) = \delta^c$, if $e=vw \in E(G_0)$ is a cut edge such that $d^{c}_{G_0-e}(v)=d^{c}_{G_0}(v)$ and $d^{c}_{G_0-e}(w)=d^{c}_{G_0}(w)$, then one can construct a properly colored tree $T$ of order at least $2\delta^c+1$ in polynomial time.
    \end{lem}
    
    \begin{proof}
    Let the two connected components of $G_0-e$ be $G_1$ and $G_2$, and assume $v\in V(G_1)$ and $w\in V(G_2)$.
    Since $d^{c}_{G_0-e}(v) = d^{c}_{G_0}(v)$, there exists an edge in $G_1$ incident to $v$ whose color is $c(vw)$. 
    Hence there exists a vertex $u\in V(G_1)$ such that $uv\in E(G_0)$ and $c(uv)=c(vw)$.
    Similarly, there exists a vertex $x\in V(G_2)$ such that $wx\in E(G_0)$ and $c(wx)=c(vw)$. Let $A_v \subseteq N_{G_1}(v)$ be a largest subset such that $u \in A_v$ and the edges between $v$ and $A_v$ have pairwise distinct colors.
    Let $A_w$ be defined analogously in $G_2$ with $x\in A_w$.
    Then $|A_v|=d^{c}_{G_1}(v)=d^{c}_{G_0}(v)\ge \delta^c$ and $|A_w|=d^{c}_{G_2}(w)=d^{c}_{G_0}(w)\ge \delta^c$.

    First, assume that $|A_v| \geq \delta^c+1$.
    We choose a subset $A'_v \subseteq A_v$ such that $|A'_v| = \delta^c+1$ and $u \in A'_v$,  as well as a subset $A'_w \subseteq A_w$ such that $|A'_w| = \delta^c$ and $x \in A'_w$.
    Let $T$ be the subgraph of $G_0$ with vertex set $(A'_v \setminus \{u\}) \cup (A'_w \setminus \{x\}) \cup \{v,w\}$ and edge set $\{va\colon a \in A'_v \setminus \{u\}\} \cup \{wb\colon b \in A'_w \setminus \{x\}\} \cup \{vw\}$.
    Then $T$ is a tree.
    The order of $T$ is $(|A'_v|-1) + (|A'_w|-1) + 2 = 2\delta^c+1$.
    The edges $va$ for $a \in A'_v \setminus \{u\}$ have pairwise distinct colors, and none of them has color $c(vw)$.
    Similarly, the edges $wb$ for $b \in A'_w \setminus \{x\}$ have pairwise distinct colors, and none of them has color $c(vw)$.
    Therefore, $T$ is a properly colored tree of order $2\delta^c+1$.
    The case $|A_w| \geq \delta^c+1$ follows by symmetry.
    
    It remains to consider the case $|A_v|=|A_w|=\delta^c$.
    Let $y\in A_v\setminus\{u\}$ arbitrary, and set $B\coloneqq (A_v\setminus\{u,y\})\cup\{v\}$. Then $|B|=\delta^c-1$. 
    The edges between $y$ and $B$ use at most $\delta^c-1$ distinct colors.
    Since $d^{c}_{G_0}(y) \geq \delta^c$, there is a vertex $z \in V(G_1) \setminus B$ such that $yz \in E(G_0)$ and $c(yz) \neq c(yv)$.
    Note that $z \notin A_v \setminus \{u\}$ and $z \neq v$.
    Let $T$ be the subgraph with vertex set $V(T) = (A_v \setminus \{u\}) \cup (A_w \setminus \{x\}) \cup \{v,w,z\}$ and edge set $E(T) = \{va\colon a \in A_v \setminus \{u\}\} \cup \{wb\colon b \in A_w \setminus \{x\}\} \cup \{vw, yz\}$.
    The graph $T$ is a tree.
    It consists of two stars centered at $v$ and $w$ joined by the edge $vw$, and an extra leaf $z$ attached to $y$.
    The order of $T$ is $2\delta^c+1$.
    Next, we verify that $T$ is properly colored.
    At $v$, the edge $vw$ has color $c(vw)$, and no edge $va$ for $a \in A_v \setminus \{u\}$ has this color.
    Similarly, no edge $wb$ for $b \in A_w \setminus \{x\}$ has  color $c(vw)$.
    The two incident edges $yv$ and $yz$ have different colors because $c(yz) \neq c(yv)$.
    All other edges in the two stars have pairwise distinct colors.
    Therefore, $T$ is a properly colored tree of order $2\delta^c+1$.
    
    All the sets defined above and the vertex $z$ can be found by checking the incident edges and their colors.
    Therefore, this construction runs in polynomial time.
    \end{proof}
    
    Consider Steps~\ref{step:while-s}-\ref{step:delete} of Algorithm~\ref{algo:properly-tree}. Observe that $\delta^c(G^\star) = \delta^c$ holds throughout the loop. 
    If the chosen edge $e$ is a cut edge then, by Lemma~\ref{lem:bridge-certificate}, the algorithm returns a properly colored tree of order at least $2\delta^c+1$.
    Otherwise, the loop terminates with a connected graph $G'$ satisfying $\delta^c(G') = \delta^c$. 
    Moreover, for every edge $e = vw \in E(G')$, deleting $e$ reduces the color degree of at least one of $v$ and $w$.
   
    \begin{cl}\label{cl:mono-star-Gstar}
    Every monochromatic component of $G'$ is a star.
    \end{cl}
    
    \begin{proof}
    Let $S$ be a monochromatic component of $G'$. Suppose first that $S$ contains a cycle, and let $e$ be an edge on this cycle. Then $G' - e$ is connected, so $e$ would be deleted by Algorithm~\ref{algo:properly-tree}, contradicting the termination of its while-loop. Hence, $S$ is a tree.

    We now show that $S$ is in fact a star. Suppose for contradiction that $S$ is not a star, and let $v_1v_2v_3v_4$ be a monochromatic path in $S$. Set $e=v_2v_3$. Then $d^c_{G'-e}(v_2)=d^c_{G'}(v_2)$ and $d^c_{G'-e}(v_3)=d^c_{G'}(v_3)$. If $G'-e$ is disconnected, then $e$ is a cut edge. By Lemma~\ref{lem:bridge-certificate}, the algorithm would output a properly colored tree of order at least $2\delta^c+1$ in Step~\ref{step:cut-end}, a contradiction. Otherwise, $G'-e$ is connected, so $e$ would again be deleted by the algorithm, contradicting termination.
    \end{proof}

    In Step~\ref{step:recolor}, we construct $\bar{G}'$ from $G'$ by assigning a distinct new color to each monochromatic component of $G'$.
    By Claim~\ref{cl:mono-star-Gstar}, every color class in $\bar{G}'$ induces a star.
    Hence, $\bar{G}'$ is star-colored and $d^c_{\bar{G}'}(v) = d^c_{G'}(v)$ for all $v \in V(\bar{G}')$.

    \begin{cl}\label{cl:pc-rainbow-equivalence}
    Let $F$ be either $G'$ or $G''$.
    Let $H$ be a subgraph of $F$, and let $\bar{H}$ be the corresponding subgraph of $\bar{F}$ with $E(\bar{H})=E(H)$.
    Then $H$ is properly colored in $F$ if and only if $\bar{H}$ is rainbow in $\bar{F}$.
    \end{cl}

    \begin{proof}
    In $\bar{F}$, two distinct edges have the same color exactly when they come from the same monochromatic component of $F$.
    Since every monochromatic component of $F$ is a star, this is equivalent to saying that the two edges are adjacent and have the same color in $F$.
    Hence $H$ is properly colored in $F$ exactly when $\bar{H}$ is rainbow in $\bar{F}$.
    \end{proof}

    Now we prove that subroutine \textsf{BuildRainbowTree} finds a rainbow tree of order $\min\{|V(G)|, 2\delta^c+1\}$.
    Recall that $\bar{G}' \notin \cG^2$, so we consider two cases according to whether $\bar{G}' \in \cG^1 \cup \cG^4 \cup \cG^6$.
    
    \begin{case}
        $\bar{G}' \notin  \cG^1 \cup \cG^4 \cup \cG^6$. 
    \end{case}
    
    By Theorem~\ref{thm:rb-star} and Theorem~\ref{thm:Rb-ex}, subroutine \textsf{BuildRainbowTree} finds a rainbow tree $T$ of order at least $\min\{|V(G)|, 2\delta^c+1\}$.
    By Claim~\ref{cl:pc-rainbow-equivalence}, the tree $T$ corresponds to a properly colored tree in $G$.
    
    \begin{case}
        $\bar{G}' \in \cG^1 \cup \cG^4 \cup \cG^6$.
    \end{case}
    
    Suppose that $\bar{G}'\in \cG^1\cup\cG^4\cup\cG^6$.
    The algorithm constructs $G''$ by choosing an edge $e_0=xy$ deleted in Step~\ref{step:delete} and an inclusion-minimal edge set $E^\bullet$.
    Note that $\delta^c(\bar{G}'') = \delta^c$.
    Therefore, it suffices to show that $\bar{G}''$ contains a rainbow tree of order $2\delta^c+1$.

    Let $H=G'\cup\{e_0\}$.
    Let $\alpha$ be the original color of $e_0$ in $H$, and let $S$ be the monochromatic component of color $\alpha$ in $H$ that contains $e_0$.
    By Claim~\ref{cl:mono-star-Gstar} and the definition of $H$, every monochromatic component of $H$ other than $S$ is a star.
    Hence, we may assume that
    $$
    E^\bullet\subseteq E(S)\setminus\{e_0\}.
    $$    
    In Step~\ref{step:recolor-again} of Algorithm~\ref{algo:properly-tree}, each monochromatic component $Q$ of $G''$ is recolored with a distinct new color $c_Q$.
    This yields a star-colored graph $\bar{G}''$.
    Let $Q_0$ be the monochromatic component of $G''$ containing $e_0$. 
    If $T_0$ is a rainbow tree in $\bar{G}''$ such that $y\notin V(T_0)$ and $E(T_0)\cap E(Q_0)=\emptyset$, then $T_0+e_0$ is a rainbow tree.
    
    \begin{subcase}
        $\bar{G}' \in \cG^6$.
    \end{subcase}

        \begin{figure}[t!]
        \centering
        \includegraphics[width=0.35\linewidth]{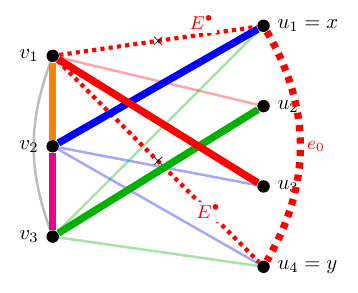}
        \caption{Illustration of Case $\bar{G}' \in \cG^6$ of the proof of Theorem~\ref{thm:main}. Take $\bar G'\cong G^6_{3,1}$. Add $e_0=u_1u_4$. The conflict edges are $E^\bullet=\{v_1u_1,v_1u_4\}$, and the thick edges and $e_0$ form the rainbow tree on $2\delta^c+1=7$ vertices.}
        \label{fig:g6}
    \end{figure}

    Assume that $\bar{G}' \cong G^6_{m,k}$ for some positive integers $m,k$.
    By the definition of  $G^6_{m,k}$, the vertex set $V(\bar{G}')$ can be partitioned into the vertex set $\{v_1, \dots, v_m\}$ of a rainbow complete graph $K_m$ and an independent set $I = \{u_1, \dots, u_{m+k}\}$.
    The bipartite graph between $V(K_m)$ and $I$ is complete.
    In particular, $\delta^c(\bar{G}')=m$, and hence $\delta^c(\bar{G}'')=m$.
    Thus, any pair of non-adjacent vertices in $G'$ both belong to $I$.
    Therefore, for $e_0=xy$, we have $x,y \in I$.
    Note that $E^\bullet$ contains no edge incident to $I\setminus\{x,y\}$. Without loss of generality, let $x=u_1$ and $y=u_{m+1}$.
    Choose $m-1$ distinct vertices $u_2,\dots,u_m\in I\setminus\{x,y\}$.
    By the above argument, for each $i$ with $2 \le i \le m$, the edge $v_i u_i$ is in $\bar{G}''$.
    
    Note that $Q_0$ is a star centered at either $x$ or $y$.
    Since $d^c_{G''}(x)=m\ge 2$, there exists an edge $xv_s\in E(G'')$ whose original color is distinct from that of $e_0$.
    By relabeling the vertices of $K_m$, we may assume that $s=1$.
    Let $e_x=v_1x$ and $P=v_1v_2\cdots v_m$ be a rainbow path in $\bar{G}''$.
    Since $K_m$ is rainbow in $\bar{G}'$ and $E^\bullet \cap E(K_m) = \emptyset$, such a path $P$ exists.
    Let $T_0$ be the graph in $\bar{G}''$ with edge set
    $$
    E(T_0)=E(P)\cup\{e_x\}\cup\{v_i u_i\colon 2\le i\le m\}.
    $$
    Observe that $|V(T_0)|=2m$.
    Furthermore, $T_0$ is rainbow in $\bar{G}''$, and $e_x\notin E(Q_0)$.
    Hence, $E(T_0)\cap E(Q_0)=\emptyset$.
    Since $y\notin V(T_0)$, the graph $T=T_0+e_0$ is a rainbow tree of order $2m+1=2\delta^c(\bar{G}'')+1$ in $\bar{G}''$; see Figure~\ref{fig:g6}.
    
    \begin{subcase}
        $\bar{G}' \in \cG^4$.
    \end{subcase}

        \begin{figure}[t!]
    \centering
    \begin{subfigure}[t]{0.48\textwidth}
        \centering
        \includegraphics[width=0.75\textwidth]{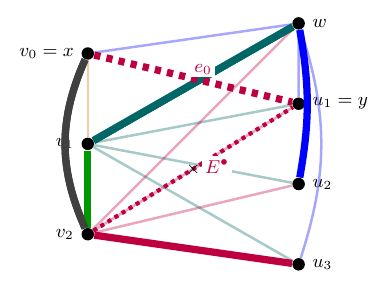}
        \caption{Case: $x=v_0$. Add $e_0=v_0u_1$. The conflict edges are $E^\bullet=\{v_2u_1\}$, and the thick edges and $e_0$ form the rainbow tree on $2\delta^c+1=7$ vertices.}
        \label{fig:g4a}
    \end{subfigure}\hfill
    \begin{subfigure}[t]{0.48\textwidth}
        \centering
        \includegraphics[width=0.75\textwidth]{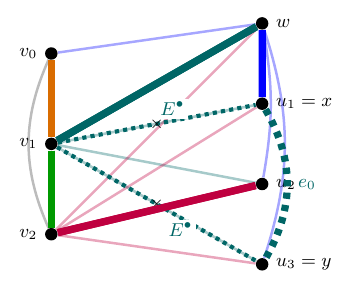}
        \caption{Case: $x\in I$. Add $e_0=u_1u_3$. The conflict edges are $E^\bullet=\{v_1u_1,v_1u_3\}$, and the thick edges and $e_0$ form the rainbow tree on $2\delta^c+1=7$ vertices.}
        \label{fig:g4b}
    \end{subfigure}\hfill
    \caption{Illustration of Case $\bar{G}' \in \cG^4$ of the proof of Theorem~\ref{thm:main}. Take $\bar G'\cong G^4_{2,1}$. }
    \label{fig:g4}
    \end{figure}
    
    Assume that $\bar{G}' \cong G^4_{m,k}$ for some positive integers $m,k$.
    Since $\delta^c(\bar{G}'') = \delta^c(\bar{G}') = m+1$, it follows that $2\delta^c+1=2m+3$.
    We construct a rainbow tree on $2m+3$ vertices in $\bar{G}''$.
    By the definition of $G_{m,k}^4$, the vertex set $V(\bar{G}')$ can be partitioned into the vertex set $\{v_0, v_1, \dots, v_m\}$ of a complete graph $C$, a vertex $w$, and an independent set $I = \{u_1, \dots, u_{m+k}\}$.
    In $G'$, two distinct vertices are non-adjacent if and only if they both belong to $I \cup \{v_0\}$.
    Therefore, the edge $e_0 = xy$ is either $u_pu_q$ for distinct vertices $u_p, u_q \in I$, or $v_0u_p$ for some vertex $u_p \in I$.
    In either case, we may assume without loss of generality that $y \in I$.
    Specifically, if $e_0=u_pu_q$, then set $x=u_p$ and $y=u_q$; if $e_0=v_0u_p$, then set $x=v_0$ and $y=u_p$.
    
    We construct a rainbow tree $T_0$ in $\bar{G}''$ such that $y\notin V(T_0)$ and $E(T_0)\cap E(Q_0)=\emptyset$.
    Since $C$ is rainbow in $\bar{G}'$, the set $E^\bullet$ contains no edge in $C$.
    Thus, $C$ is also a rainbow clique in $\bar{G}''$.
    Therefore, $C$ has a rainbow spanning tree $T_C$.
    If $Q_0$ contains an edge of $C$, then $E(Q_0) \cap E(C)$ consists of exactly one edge.
    Choose $T_C$ such that it does not contain this edge.

    We expand $T_C$ by adding the edge $v_1w$.
    Choose $m$ distinct vertices $u_{j_1},\dots,u_{j_m} \in I\setminus\{y\}$.
    We add the edge $wu_{j_1}$ and the edges $v_i u_{j_i}$ for each $i$ with $2 \le i \le m$.
    Let $T_0$ denote the resulting graph.
    Then $T_0$ is a tree with the vertex set $V(C)\cup\{w\}\cup\{u_{j_1},\dots,u_{j_m}\}$.
    Thus, $|V(T_0)|=(m+1)+1+m=2m+2$.

    If $x=v_0$, then by construction, $T_0$ contains no edge incident to $v_0$ outside of $T_C$.
    Thus $E^\bullet \cap E(T_0) = \emptyset$.
    Then $T_0$ is a tree in $\bar{G}''$.
    Since $T_C$ was chosen to avoid $E(Q_0)\cap E(C)$, we have $E(T_0)\cap E(Q_0)=\emptyset$.
    If $x \in I$, then choose the vertices $u_{j_1}, \dots, u_{j_m}$ so that $x$ is one of them.
    In $T_0$, the vertex $x$ has exactly one incident edge.
    We have $\delta^c(G'') = m+1$.
    All edges of $Q_0$ have the original color of $e_0$.
    Thus, there is an edge $e$ incident to $x$ in $G''$ such that the original color of $e$ is different from that of $e_0$.
    We choose $e$ to be the unique edge incident to $x$ in $T_0$.
    This implies that $E(T_0) \cap E(Q_0) = \emptyset$ and $E^\bullet \cap E(T_0) = \emptyset$.
    Thus $T_0$ is a tree in $\bar{G}''$.
    
    Finally, let $T=T_0+e_0$.
    Since $y\notin V(T_0)$ and $E(T_0)\cap E(Q_0)=\emptyset$, the graph $T$ is a rainbow tree of order $2\delta^c(\bar{G}'')+1$ in $\bar{G}''$; see Figure~\ref{fig:g4}.
    
    \begin{subcase}
        $\bar{G}' \in \cG^1$.
    \end{subcase}

        \begin{figure}[t!]
    \centering
    \begin{subfigure}[t]{0.48\textwidth}
        \centering
        \includegraphics[width=0.85\textwidth]{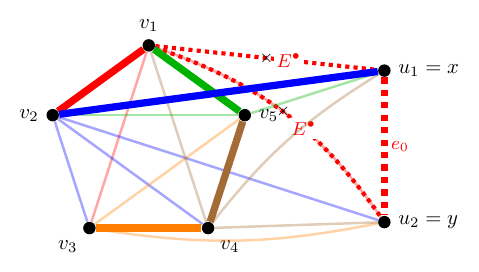}
        \caption{Case: $x\in I$. Add $e_0=u_1u_2$. The conflict edges are $E^\bullet=\{v_1u_2\}$, and the thick edges and $e_0$ form the rainbow tree on $2\delta^c+1=7$ vertices.}
        \label{fig:g1a}
    \end{subfigure}\hfill
    \begin{subfigure}[t]{0.48\textwidth}
        \centering
        \includegraphics[width=0.85\textwidth]{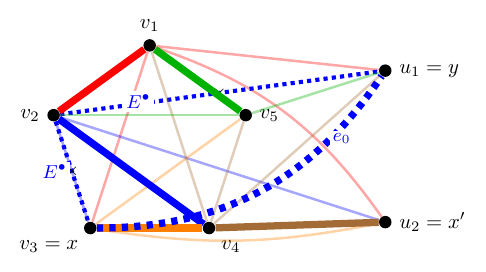}
        \caption{Case: $x\in V(K_m)$.Add $e_0=u_1v_3$. The conflict edges are $E^\bullet=\{v_2v_3\}$, and the thick edges and $e_0$ form the rainbow tree on $2\delta^c+1=7$ vertices.}
        \label{fig:g1b}
    \end{subfigure}\hfill
    \caption{Illustration of Case $\bar{G}' \in \cG^1$ of the proof of Theorem~\ref{thm:main}.}
    \label{fig:g1}
    \end{figure}

    Assume that $\bar{G}' \cong G^1_{m,k}$ for some odd integer $m$ and an integer $k \geq 2$.
    By the definition of $G^1_{m,k}$, the vertex set $V(\bar{G}')$ is partitioned into the vertex set $\{v_1, \dots, v_m\}$ of a complete graph $K_m$ and an independent set $I = \{u_1, \dots, u_k\}$.
    Define an orientation $D$ of $K_m$ by directing the edge $v_i v_j$ from $v_i$ to $v_j$ whenever $c(v_i v_j)$ is the color center at $v_i$. Note that $D$ is a regular tournament on $m\ge 5$ vertices.

    \begin{cl}\label{cl:strong-connected-e}
    For every arc $e\in A(D)$, the digraph $D-e$ is strongly connected.
    \end{cl}
    
    \begin{proof}
    Suppose for a contradiction that $D-e$ is not strongly connected.
    Then there is a partition of $V(D)$ into two nonempty sets $A$ and $B$ such that all arcs between $A$ and $B$ are directed from $A$ to $B$.
    Let $a = |A|$ and $b = |B|$.
    
    Since $D$ is a tournament, in $D$ there are either $ab$ or $ab-1$ arcs from $A$ to $B$, according as $e$ is not, or is, the unique arc from $B$ to $A$.
    Hence
    \[
    a\cdot \frac{m-1}{2}
    =\sum_{v\in A} d_D^+(v)
    =\frac{a(a-1)}{2}+ab-t,
    \]
    where $t\in\{0,1\}$.
    Multiplying by $2/a$ gives $m-1=a-1+2b-\frac{2t}{a}$. Using $m=a+b$, we obtain $b=\frac{2t}{a}$.
    If $t=0$, then $b=0$, a contradiction.
    So $t=1$, and then $ab=2$.
    Since $a,b\ge 1$ and $a+b=m\ge 5$, we have $ab\neq 2$.
    This contradiction proves that $D-e$ is strongly connected.
    \end{proof}
    
    Note that $\delta^c(\bar{G}') = \frac{m+1}{2}$.
    Thus, $2\delta^c(\bar{G}'')+1 = m+2$.
    For any two nonadjacent vertices in $G'$, at least one lies in $I$.
    Thus, for the added edge $e_0=xy$, we may assume that $y \in I$.

    We construct a rainbow tree $T_0$ of order $m+1$ in $\bar{G}''$ such that $y \notin V(T_0)$ and $E(T_0) \cap E(Q_0) = \emptyset$.
    Then $T = T_0+e_0$ is a rainbow tree of order $m+2 = 2\delta^c(\widehat{G'})+1$.
    
    First, assume that $x \in I$.
    We have $d^c_{G''}(x) \geq \delta^c(G'') \geq 2$.
    Thus, there is a vertex $v_t \in V(K_m)$ such that $v_t x \in E(G'')$ and $v_t x \notin E(Q_0)$.
    Fix such a vertex $v_t$.
    By Claim~\ref{cl:core-tree-avoid}, there is a rainbow spanning tree $T_t$ in $K_m$ that contains no color $c(v_t x)$. Let $T_0 = T_t + v_t x$.
    Then $T_0$ is a tree of order $m+1$, and $y \notin V(T_0)$. 
    By the choice of $v_t$, we have $E(T_0) \cap E(Q_0) = \emptyset$.
    
    Now, assume that $x \in V(K_m)$. 
    Since $k \ge 2$, we may choose a vertex $x' \in I \setminus \{y\}$.
    Choose $t \in [m]$ such that $v_t x' \in E(G'')$.
    If $c(e_0)$ is the color center at $x$, then $E^\bullet = \emptyset$, contradicting Step~\ref{step:delete} of Algorithm~\ref{algo:properly-tree}. 
    Thus, there is a unique edge $e$ in $E(K_m)\cap E^{\bullet}$.
    By Claim~\ref{cl:strong-connected-e} and similar argument with the proof of Claim~\ref{cl:core-tree-avoid}, the graph $K_m - e$ admits a rainbow spanning tree $T_t$ that contains no color $c(v_t x')$. Let $T_0 = T_t + v_t x'$.
    Then $T_0$ is a rainbow tree of order $m+1$. Moreover, $E(T_0) \cap E(Q_0) = \emptyset$.

    Finally, let $T=T_0+e_0$.
    Then $T$ is a rainbow tree of order $2\delta^c(\bar{G}'')+1$ in $\bar{G}''$; see Figure~\ref{fig:g1}.
    
    \medskip
    
    In all cases, subroutine \textsf{BuildRainbowTree} finds a rainbow tree of order at least $2\delta^c+1$.
    By Claim~\ref{cl:pc-rainbow-equivalence}, this tree corresponds to a properly colored tree of order at least $\min\{|V(G)|, 2\delta^c+1\}$ in $G$.
    Each step of the algorithm takes polynomial time.
    Therefore, the overall algorithm runs in polynomial time.
    This completes the proof of Theorem~\ref{thm:main}.
\end{proof}

\section{Conclusions}
\label{sec:conclusions}

In this paper, we studied the above-guarantee problem for properly colored trees in connected edge-colored graphs. We also presented a polynomial-time algorithm that, given a connected graph $G$, constructs a properly colored tree of order at least $\min\{|V(G)|,2\delta^c(G)+1\}$ whenever such a tree exists.

We close the paper by mentioning some open problems:

\begin{enumerate}\itemsep0em
    \item The most straightforward question is whether there exists an FPT algorithm, parameterized by $k$, that finds a properly colored tree of order at least $\min\{|V(G)|,2\delta^c(G)+k\}$ in a connected graph $G$. A obstacle to a general above-guarantee FPT result already appears in the preprocessing step. In our proof for the case $k=1$, whenever the deletion process encounters a color degree redundant cut edge, Lemma~\ref{lem:bridge-certificate} gives a local construction of a properly colored tree of order $2\delta^c(G)+1$, so the algorithm can stop immediately. 
    For a target $2\delta^c(G)+k$ with $k\ge 2$, the same cut edge certificate only provides the first extra vertex and gives no mechanism for obtaining the remaining $k-1$ vertices.
    \item Following the framework of Fomin et al.~\cite{fomin2024approxdirac}, a natural question is to study the above-guarantee approximability of the maximum properly colored tree problem.
\end{enumerate}

\medskip
\paragraph{Acknowledgement.}
Yuhang Bai was supported by the National Natural Science Foundation of China (12131013 and 12471334), by China Scholarship Council (202406290002), and by Shaanxi Fundamental Science Research Project for Mathematics and Physics (22JSZ009). The research received further support from the Lend\"ulet Programme of the Hungarian Academy of Sciences (LP2021-1/2021), from the Ministry of Innovation and Technology of Hungary from the National Research, Development and Innovation Fund (ADVANCED 150556, ADVANCED 153096, and ELTE TKP 2021-NKTA-62), and from the Dynasnet European Research Council Synergy project (ERC-2018-SYG 810115).

\bibliographystyle{abbrv}
\bibliography{ref}

\appendix

\section{Remaining cases in the proof of Theorem~\ref{thm:rb-star}}
\label{app:m1m2}

In this appendix we continue the proof of Theorem~\ref{thm:rb-star} in the remaining cases
$m=2$ and $m=1$ under the standing assumptions
\[
|V(T)|=2\delta^c(G)\quad\text{and}\quad V(T)\neq V(G).
\]
We keep all notation from the main proof:
the set of reachable colors $R=\{c_1,\dots,c_r\}$, the components $W_1,\dots,W_{r+1}$ of $G[V(T)]\setminus (E_{c_1}\cup\cdots\cup E_{c_r})$, the sets $C(v)$ and $S$, the integer $m$, and the representatives $x_i\in V(W_i)\cap S$ for $i\in[m]$, as well as $\beta_i = |(S\setminus\{x_1,\dots,x_m\})\cap V(W_i)|$. 
We also use Claims~\ref{cl:rbmt} and \ref{cl:adjacency-tight} from the main proof.

\subsection{Case \texorpdfstring{$m=2$}{m=2}}
\label{app:m2}

Assume that $m=2$.
By equality in \eqref{eq:deg-bound-single}, we have
$d_G^c(x_1)=d_G^c(x_2)=\delta^c(G)$.
Moreover, tightness in the counting step yielding $r+1$ components implies
\begin{equation}\label{eq:m2-singletons}
|W_3|=\cdots=|W_{r+1}|=1.
\end{equation}

\begin{cl}\label{cl:m2-one-component-inside-S}
Either $V(W_1)\subseteq S$ or $V(W_2)\subseteq S$.
\end{cl}
\begin{proof}
Suppose for a contradiction that there exist vertices $w_1\in V(W_1)\setminus S$ and $w_2\in V(W_2)\setminus S$.
Then
\[
d_G^c(w_i)\le |W_i|-1+|N_G(w_i)\cap(S\setminus V(W_i))|
\]
for $i=1,2$. Using \eqref{eq:m2-singletons}, we get
\[
d_G^c(w_1)+d_G^c(w_2)
\le (|W_1|+|W_2|-2)+r
= \bigl(|W_1|+|W_2|+(r-1)\bigr)-1
= |V(T)|-1
=2\delta^c(G)-1,
\]
contradicting $d_G^c(w_1)+d_G^c(w_2)\ge 2\delta^c(G)$.
\end{proof}

By symmetry, we may assume that $V(W_1)\subseteq S$.

\begin{cl}\label{cl:m2-if-big-intersection-then-other-small}
If $|V(W_i)\cap S|\ge2$ for some $i\in\{1,2\}$, then $|V(W_{3-i})\cap S|=1$.
\end{cl}
\begin{proof}
Assume $|V(W_i)\cap S|\ge2$ and $|V(W_{3-i})\cap S|\ge2$.
Choose $y\in (V(W_{3-i})\cap S)\setminus\{x_{3-i}\}$.
By Claim~\ref{cl:adjacency-tight}, $x_i y\in E(G)$ and $c(x_i y)\notin C(x_i)$.
Since $x_i$ and $y$ lie in different components of $G[V(T)]-\bigcup_{c\in R}E_c$, the edge $x_i y$ has a reachable color, and hence it belongs to one of its endpoints.
Since $c(x_i y)\notin C(x_i)$, this color belongs to $y$.

The vertices $x_{3-i}$ and $y$ play symmetric roles in $W_{3-i}$.
Therefore, we may replace $x_{3-i}$ by $y$.
By Claim~\ref{cl:adjacency-tight} and $|V(W_i) \cap S| \ge 2$, we have that $c(x_i x_{3-i})$ belongs to $x_i$, a contradiction.
\end{proof}

\begin{cl}\label{cl:m2-W1-singleton}
$V(W_1)=\{x_1\}$.
\end{cl}
\begin{proof}
Suppose $|W_1|\ge2$. Then $|V(W_1)\cap S|\ge2$, and Claim~\ref{cl:m2-if-big-intersection-then-other-small} yields
$|V(W_2)\cap S|=1$, i.e.\ $V(W_2)\cap S=\{x_2\}$ and hence $\beta_2=0$ and $\beta_1=|W_1|-1$.
And Claim~\ref{cl:adjacency-tight} gives $|C(x_1)|=1$ and $c(x_1x_2)\in C(x_1)$.
Therefore, equality in \eqref{eq:deg-bound-single} for $x_1$ yields
\[
d_G^c(x_1)=|C(x_1)|+|W_1|-1=|W_1|.
\]
For $x_2$, the same equality yields
\[
d_G^c(x_2)=|C(x_2)|+|W_2|-1+|W_1| \ge 1+0+|W_1|=|W_1|+1,
\]
contradicting $d_G^c(x_1)=d_G^c(x_2)=\delta^c(G)$.
\end{proof}

\begin{subcase}
    $V(W_2)\subseteq S$.
\end{subcase}

Then the same argument as above implies that $V(W_2)=\{x_2\}$, hence $S=\{x_1,x_2\}$.
Let $w$ be the unique vertex of $W_3$.
Then $w$ has no neighbors in $V(G)\setminus V(T)$, and $N_G(w)\subseteq \{x_1,x_2\}$, so $d_G^c(w)\le2$.
By $|V(T)|=2\delta^c(G)\ge3$, we have $\delta^c(G)\ge2$, so $\delta^c(G)=2$ and $|V(T)|=4$.
Therefore $r+1=4$ and $W_3,W_4$ are singletons, each of them is adjacent to both $x_1$ and $x_2$.
This is exactly the structure of the family $\cG^2$.

\begin{subcase}
    $V(W_2)\not\subseteq S$.
\end{subcase}
Choose $w_2\in V(W_2)\setminus S$.

\begin{cl}\label{cl:m2-G4-structure}
If $V(W_2)\not\subseteq S$, then the following statements hold.
\begin{enumerate}[label=(\alph*)]\itemsep0em
\item For any $i$ and $j$ with $1 \le i < j \le r$, the colors $c_i$ and $c_j$ belong to distinct vertices of $S$. \label{it:a}
\item $V(W_2)\setminus S=\{w_2\}$.\label{it:b}
\item $W_2$ is a rainbow clique, and $E(W_2)$ uses no color from $R$.\label{it:c}
\item For every $x\in V(W_2)\cap S$, the edge $xx_1$ has a color in $C(x)$.\label{it:d}
\item Every vertex in $W_3\cup\cdots\cup W_{r+1}$ is adjacent to every vertex in $S$.\label{it:e}
\item $x_1w_2\in E(G)$ and $c(x_1w_2)\in C(x_1)$.\label{it:f}
\end{enumerate}
\end{cl}
\begin{proof}
\ref{it:a} If some $v\in S$ is the center of two reachable colors, then $r\ge |S|+1$.
By \eqref{eq:m2-singletons}, $|W_3\cup\cdots\cup W_{r+1}|=r-1\ge |S|$, and by $W_2\not\subseteq S$, we have $|W_1\cup W_2|\ge |S|+1$.
Hence $|V(T)|\ge 2|S|+1$.
Since $d_G^c(w)\le |S|$ for $w\in V(W_3)$, we have $\delta^c(G)\le |S|$, contradicting $|V(T)|=2\delta^c(G)$.

\ref{it:b} If $|V(W_2)\setminus S|\ge2$, then (a) implies $|W_2|\ge (r-1)+2=r+1$.
Pick $x\in V(W_2)\cap S$ and use Claim~\ref{cl:adjacency-tight} with $x$ chosen as the representative in $V(W_2)\cap S$, we get $d_G^c(x)\ge |C(x)|+|W_2|-1=|W_2|\ge r+1$.
But every vertex in $W_3$ has color degree at most $|S|=r$, so $\delta^c(G)\le r$, contradiction.

\ref{it:c} Since we may choose any vertex of $V(W_2)\cap S$ as $x_2$, Claim~\ref{cl:adjacency-tight} implies that $G[V(W_2)\cap S]$ is a rainbow clique.
By the definition of the components $W_i$, $E(W_2)$ uses no color from $R$.
Finally, if there are two edges of the same color incident to $w_2$, then $d_G^c(w_2)\le |W_2|-1-1+1<\delta^c(G)$, a contradiction.
Therefore, $W_2$ is a rainbow clique.

\ref{it:d} If $|V(W_2)\cap S|\ge2$, it holds by Claim~\ref{cl:adjacency-tight}.
If $V(W_2) \cap S = \{x_2\}$, then by \ref{it:a}, $r = |S| = 2$, which implies $\delta^c(G) = 2$. 
If $c(x_1x_2) \in C(x_1)$, then $x_2$ has at least three distinct colors on its incident edges: $c(x_1x_2)$, the unique color in $C(x_2)$, and $c(x_2w_2)$.
This contradicts $d_G^c(x_2) = 2$.

\ref{it:e} By \ref{it:a} and \ref{it:b}, $|W_2|=(r-1)+1=r$, so $|V(T)|=|W_1|+|W_2|+(r-1)=1+r+(r-1)=2r$, hence $\delta^c(G)=r$.
For every vertex $w\in V(W_3)\cup\cdots\cup V(W_{r+1})$, we have $d_G^c(w)\le |S|=r$, and hence $w$ is adjacent to all of $S$.

\ref{it:f} By \ref{it:c} we have that $w_2$ is incident inside $W_2$ with $|W_2|-1=r-1$ distinct colors.
Since $\delta^c(G)=r$, we have $x_1w_2\in E(G)$ and $c(x_1w_2)\in C(x_1)$.
\end{proof}

\begin{cl}\label{cl:m2-outside-independent}
The set $V(G)\setminus V(T)$ is independent.
\end{cl}

\begin{proof}
Assume for a contradiction that $V(G)\setminus V(T)$ is not independent.
Since $G$ is connected, there exist adjacent vertices $x,y\in V(G)\setminus V(T)$ such that
$x$ is adjacent to some vertex $z\in V(T)$.
By Claim~\ref{cl:rbmt}, the color of the edge $xz$ lies in $R$, and $z\in S$.
Since every monochromatic subgraph of $G$ is a star, we have $c(xy)\neq c(xz)$ and $c(xy)\notin R$.

Let $w$ be the unique vertex of $W_{r+1}$.
By the previous part of the proof, $W_3,\dots,W_{r+1}$ are singleton components.
Moreover, every vertex of $W_3\cup\cdots\cup W_{r+1}$ is adjacent to every vertex of $S$, and the colors on these edges are exactly the colors in $R$.
Since $|V(T)|=2\delta^c(G)=2r$ and $W_2\not\subseteq S$, we have $r\ge 2$.
Choose a color $c_i\in R\setminus \{c(xz)\}$.

By Claim~\ref{cl:m2-G4-structure}, the graph $G[V(W_1)\cup V(W_2)]$ has a rainbow spanning tree
$T_i$ such that $c_i$ is the only color from $R$ appearing in $T_i$.
For each color in $R\setminus\{c_i,c(xz)\}$, choose one edge of this color joining its center in $S$ to a distinct vertex among $W_3,\dots,W_r$.
Together with $T_i$, these edges form a rainbow tree $T^\ast$ on $V(T)\setminus\{w\}$.

Since $c(xy)\notin R$ and every monochromatic subgraph is a star, the color $c(xy)$ does not appear on any edge of $T^\ast$.
Hence
\[
T^\ast + xz + xy
\]
is a rainbow tree on the vertex set $(V(T)\setminus\{w\})\cup\{x,y\}$.
This contradicts Steps~\ref{step:exchange-one-edge}--\ref{step:exchange-two-edge} of Algorithm~\ref{algo:rainbow-tree}.
\end{proof}

Since $\delta^c(G)=|S|=r$, every vertex in $V(G)\setminus V(T)$ is adjacent to all vertices of $S$.
It is easy to see that the set $V(G)\setminus V(T)$ is independent.
Then together with Claim~\ref{cl:m2-outside-independent}, this is exactly the definition of $\cG^4$.
This completes the proof for case $m=2$.

\subsection{Case \texorpdfstring{$m=1$}{m=1}}
\label{app:m1}

Assume that $m=1$. Then $S\subseteq V(W_1)$, and $W_2,\dots,W_{r+1}$ are disjoint from $S$.

\begin{cl}\label{cl:m1-W1-equals-S}
$V(W_1)=S$.
\end{cl}
\begin{proof}
Suppose that there exists $u\in V(W_1)\setminus S$, and choose any $v\in V(W_2)$.
Since $u\notin S$, we have $d_G^c(u)\le |W_1|-1$.
For $v\in W_2$, we have $d_G^c(v)\le |W_2|-1+r$.
Thus
\[
d_G^c(u)+d_G^c(v)\le (|W_1|-1)+(|W_2|-1+r)\le |V(T)|-1=2\delta^c(G)-1,
\]
contradicting $d_G^c(u)+d_G^c(v)\ge 2\delta^c(G)$.
\end{proof}

\begin{cl}\label{cl:m1-structure}
The following statements hold.
\begin{enumerate}[label=(\alph*)]\itemsep0em
\item $|W_2|=\cdots=|W_{r+1}|=1$. \label{it:aa}
\item Every vertex in $V(T)\setminus V(W_1)$ is adjacent to every vertex of $W_1$.\label{it:bb}
\item $d_G^c(x)=\delta^c(G)$ for every $x\in V(T)$.\label{it:cc}
\item For any $i$ and $j$ with $1 \le i < j \le r$, the colors $c_i$ and $c_j$ belong to distinct vertices of $S$.\label{it:dd}
\item $W_1$ is a rainbow clique and uses no color from $R$.\label{it:ee}
\end{enumerate}
\end{cl}
\begin{proof}
Since $G$ is connected and $W_i$ is a component of $G[V(T)]-\bigcup_{c\in R}E_c$ for $i\ge2$, each $W_i$ contains a vertex $w$ adjacent to $W_1$.
Fix such a vertex $w\in V(W_i)$ and a neighbor $w_1\in V(W_1)$. As before, $d_G^c(w_1)\le |W_1|-1+|C(w_1)|$ and $d_G^c(w)\le |W_i|-1+(r-|C(w_1)|+1)$.
Then we have
\[
2\delta^c(G)\le d_G^c(w_1)+d_G^c(w)\le |W_1|+|W_i|+r-1\le |V(T)|=2\delta^c(G),
\]
so all inequalities are equalities.

From $|V(T)|=|W_1|+|W_i|+r-1$ and $\sum_{j\neq i,\,2\le j\le r+1}|W_j|\ge r-1$ we get \ref{it:aa}.
By $d_G^c(w) = |W_i|-1+(r-|C(w_1)|+1)$, $w$ is adjacent to every vertex of $W_1$, proving \ref{it:bb}.
By \ref{it:bb}, we can choose any pair of vertices from $V(W_1)$ and $V(T) \setminus V(W_1)$ as $w_1$ and $w$, respectively. Since $d_G^c(w_1)+d_G^c(w)=2 \delta^c(G)$, we obtain $d_G^c(w_1)=d_G^c(w)=\delta^c(G)$. Hence, c\ref{it:cc} holds.
For \ref{it:dd}, if some vertex in $S=V(W_1)$ is the center of two colors, then choosing $w_1\in V(W_1)$ and $w\in V(T)\setminus V(W_1)$
would contradict \ref{it:cc}.
Hence all reachable colors have distinct centers in $S$, so $|S|=r$ and $|C(v)|=1$ for $v\in S$.
Finally, By $d_G^{c}(w_1)=|W_1|-1+|C(w_1)|$ for each $w_1 \in V(W_1)$ and \ref{it:bb}, $w_1$ is adjacent to all vertices in $V(W_1) \setminus \{w_1\}$ and colors of these incident edges are distinct and not contained in $C(w_1)$. 
Since every monochromatic subgraph is a star, $W_1$ is a rainbow clique and uses no color from $R$, proving \ref{it:ee}.
\end{proof}

By Claim~\ref{cl:m1-structure}\ref{it:aa} and \ref{it:dd} and $|V(T)|=|W_1|+\sum_{i=2}^{r+1}|W_i|=r+r=2r$, we get $\delta^c(G)=r$.

\begin{cl}\label{cl:m1-outside-independent}
The set $V(G)\setminus V(T)$ is independent.
\end{cl}
\begin{proof}
Assume for a contradiction that $V(G)\setminus V(T)$ is not independent.
Since $G$ is connected, there exist adjacent vertices $x,y\in V(G)\setminus V(T)$ such that
$x$ is adjacent to a vertex $z\in V(T)$.
By Claim~\ref{cl:rbmt}, we have $c(xz)\in R$ and $z\in S=V(W_1)$, and since each color class is a star, $c(xy)\neq c(xz)$.
By Claim~\ref{cl:m1-structure}\ref{it:bb} and \ref{it:dd}, we may choose $r-1$ independent edges between $V(W_1)$ and $V(W_2)\cup\cdots\cup V(W_r)$ whose colors are pairwise distinct and different from $c(xz)$.
Together with $xz$ and $xy$, we get a rainbow tree on $(V(T)\setminus V(W_{r+1}))\cup\{x,y\}$, contradicting the termination of Algorithm~\ref{algo:rainbow-tree}.
\end{proof}

Since $\delta^c(G)=|S|=r$, every vertex in $V(G)\setminus V(T)$ is adjacent to every vertex of $W_1$.
Together with Claims~\ref{cl:m1-structure} and \ref{cl:m1-outside-independent}, this is exactly the definition of $\cG^6$.
This completes the proof for $m=1$.
\end{document}